\documentclass[11pt]{amsart}

\usepackage{amsmath}
\usepackage{amsfonts}
\usepackage{amssymb}
\usepackage{amsthm}
\usepackage{mathptmx}
\usepackage{mathrsfs}
\usepackage{latexsym}
\usepackage{times}
\usepackage[mathscr]{euscript}
\usepackage[isolatin]{inputenc}

\usepackage[pagebackref]{hyperref}
\usepackage[backrefs,alphabetic,initials]{amsrefs}

\DeclareMathAlphabet{\mathpzc}{OT1}{pzc}{m}{it}

\newtheorem{thm}{Theorem}[section]
\newtheorem{lem}[thm]{Lemma}
\newtheorem{prop}[thm]{Proposition}
\newtheorem{cor}[thm]{Corollary}

\theoremstyle{definition}
\newtheorem{defn}[thm]{Definition}
\newtheorem{ex}[thm]{Example}
\newtheorem*{aknow}{Acknowledgments}

\theoremstyle{remark}
\newtheorem{rem}[thm]{Remark}

\usepackage{color}

\newcommand{\zk}{\mathfrak{z}}

\newcommand{\g}{\mathfrak{g}}
\newcommand{\N}{\mathfrak{N}}
\newcommand{\J}{\mathfrak{J}}

\newcommand{\hk}{\mathfrak{h}}
\newcommand{\lk}{\mathfrak{l}}

\newcommand{\ak}{\mathfrak{a}}

\newcommand{\qk}{\mathfrak{q}}
\newcommand{\tk}{\mathfrak{t}}

\newcommand{\Sb}{\mathsf{S}}

\newcommand{\Zb}{\mathsf{Z}}

\newcommand\RR{\mathbb R}
\newcommand\CC{\mathbb C}

\newcommand\PP{\mathbb P}

\newcommand{\Tr}{\operatorname{Tr}}

\newcommand{\ad}{\operatorname{ad}}

\newcommand{\im}{\operatorname{Im}}

\renewcommand\dfrac{\displaystyle \frac}

\newcommand\e{\varepsilon}

\newcommand{\End}{\operatorname{End}}

\newcommand{\Id}{\operatorname{Id}}

\renewcommand\hat\widehat
\renewcommand\tilde\widetilde 
\newcommand{\spa}{\operatorname{span}}

\newcommand{\OO}{\operatorname{O}}

\newcommand{\rank}{\operatorname{rank}}

\newcommand{\Jb}{\overline{\J}}

\newcommand{\Ann}{\operatorname{Ann}}
\newcommand{\RAnn}{\operatorname{RAnn}}
\newcommand{\LAnn}{\operatorname{LAnn}}

\newcommand{\oplusp}{{ \ \overset{\perp}{\mathop{\oplus}} \ }}

\newcommand{\ps}{\PP^1}

\begin{document}

\title[Jordanian double extensions and symmetric Novikov
algebras]{Jordanian double extensions of a quadratic vector space and
  symmetric Novikov algebras}

\author{Minh Thanh Duong, Rosane Ushirobira}

\address{Institut de Mathématiques de Bourgogne, Université de
  Bourgogne, B.P. 47870, F-21078 Dijon Cedex, France}

\email{Thanh.Duong@u-bourgogne.fr}
\email{Rosane.Ushirobira@u-bourgogne.fr}



\date{\today}

\begin{abstract}
  First, we study pseudo-Euclidean Jordan algebras obtained as double
  extensions of a quadratic vector space by a one-dimensional
  algebra. We give an isomorphic characterization of 2-step nilpotent
  pseudo-Euclidean Jordan algebras. Next, we find a Jordan-admissible
  condition for a Novikov algebra $\N$. Finally, we focus on the case
  of a symmetric Novikov algebra and study it up to dimension 7.
\end{abstract}

\maketitle

\section{Introduction}

All algebras considered in this paper are finite-dimensional algebras
over $\CC$. The general framework for our study is the following: let
$\qk$ be a complex vector space equipped with a non-degenerate
bilinear form $B_\qk$ and $C:\qk\rightarrow \qk$ be a linear map. We
associate a vector space
\[\J = \qk\oplusp\tk \] to the triple $(\qk,B_\qk,C)$ where $(\tk =
\spa\{x_1,y_1\},B_\tk)$ is a 2-dimensional vector space and
$B_\tk:\tk\times\tk\rightarrow \CC$ is the bilinear form defined by
\[ B_\tk(x_1, x_1) = B_\tk(y_1, y_1) = 0, \ B_\tk(x_1, y_1) = 1.\]
Define a product $\star$ on the vector space $\J$ such that $\tk$ is a
subalgebra of $\J$, \[y_1\star x = C(x), \ x_1\star x = 0,\ x\star y =
B_\qk(C(x),y)x_1 \] for all $x,y\in\qk$ and such that the bilinear
form $B_\J = B_\qk + B_\tk$ is {\em associative} (that means $B_\J
(x\star y, z) = B_\J(x,y\star z), \ \forall x,y,z\in\J$). We call $\J$
is a {\em double extension of $\qk$ by $C$}. It can be completely
characterized by the pair $(B_\qk,C)$ combined with some properties of
the 2-dimensional subalgebra $\tk$.

A rather interesting note is that such algebras $\J$ can also be
classified up to isometric isomorphisms (or i-isomorphisms, for short)
or isomorphisms. This is successfully done for the case of $B_\qk$
symmetric or skew-symmetric, $C$ skew-symmetric (with respect to
$B_\qk$) and $B_\tk$ symmetric (see \cite{FS87}, \cite{DPU10} and
\cite{Duo10}). In these cases, a double extension of $\qk$ by $C$ is a
quadratic Lie algebra or a quadratic Lie superalgebra. Their
classification is connected to the well-known classification of
adjoint orbits in classical Lie algebras theory\cite{CM93}. That is,
there is a one-to-one correspondence between isomorphic classes of
those algebras and adjoint $G$-orbits in $\ps(\g)$, where $G$ is the
isometry group of $B_\qk$ and $\ps(\g)$ is the projective space
associated to the Lie algebra $\g$ of $G$. Therefore, it is natural to
consider similar algebras corresponding to the remaining different
cases of the pair $(B_\qk,C)$.

Remark that the above definition of a double extension is a special
case of a one-dimensional extension in terms of the double extension
notion initiated by V. Kac to construct quadratic solvable Lie
algebras \cite{Kac85}. This notion is generalized effectively for
quadratic Lie algebras \cite{MR85} and many other non-anticomutative
algebras (see \cite{BB99}, \cite{BB08} and \cite{AB10}) to obtain an
inductive characterization (also called {\em generalized double
  extension}). Unfortunately, the classification (up to isomorphisms
or i-isomorphisms) of the algebras obtained by the double extension or
generalized double extension method seems very difficult, even in
nilpotent or low dimensional case. For example, nilpotent
pseudo-Euclidean Jordan algebras up to dimension 5 are listed
completely but only classified in cases up to dimension 3 \cite{BB08}.


In Section 2, we apply the work of A. Baklouti and S. Benayadi in
\cite{BB08} for the case of a one-dimensional double extension of the
pair $(B_\qk,C)$ to obtain pseudo-Euclidean (commutative) Jordan
algebras (i.e. Jordan algebras endowed with a non-degenerate
associative symmetric bilinear form). Consequently, the bilinear forms
$B_\qk$, $B_\tk$ are symmetric, $C$ must be also symmetric (with
respect to $B_\qk$) and the product $\star$ is defined by:
\begin{eqnarray*}
  &&(x + \lambda x_1 + \mu y_1)\star( y + \lambda' x_1 + \mu' y_1)  := \\ &&\mu C(y) + \mu'C(x) + B_\qk( C(x), y) x_1 + \e \left( \left( \lambda\mu' + \lambda'\mu \right) x_1 + \mu \mu' y_1\right),
  \end{eqnarray*}
  $\e \in \{0,1 \}$, for all $x, y \in \qk, \lambda, \mu, \lambda',
  \mu' \in \CC$.

  Since there exist only two one-dimensional Jordan algebras, one
  Abelian and one simple, then we have two types of extensions called
  respectively {\em nilpotent double extension} and {\em
    diagonalizable double extension}.  The first result (Proposition
  \ref{prop2.1}, Corollary \ref{cor2.2}, Corollary \ref{cor2.7} and
  Appendix) is the following:

  \smallskip

THEOREM 1:{\em
\begin{enumerate}
\item If $\J$ is the nilpotent double extension of $\qk$ by $C$ then
  $C^3 = 0$, $\J$ is 3-step nilpotent and $\tk$ is an Abelian
  subalgebra of $\J$.

\item If $\J$ is the diagonalizable double extension of $\qk$ by $C$
  then $3C^2 = 2C^3 +C$, $\J$ is not solvable and $\tk \star \tk
  =\tk$. In the reduced case, $y_1$ acts diagonalizably on $\J$ with
  eigenvalues $1$ and $\frac{1}{2}$.
\end{enumerate}}

In Propositions \ref{prop2.3} and \ref{prop2.8}, we characterize these
extensions up to i-isomorphisms, as well as up to isomorphisms and
obtain the classification result:

\smallskip

THEOREM 2:{\em

\begin{enumerate}
\item Let $\J = \qk \oplusp (\CC x_1 \oplus \CC y_1)$ and $\J' = \qk
  \oplusp (\CC x'_1 \oplus \CC y'_1)$ be nilpotent double extensions
  of $\qk$ by symmetric maps $C$ and $C'$ respectively. Then there
  exists a Jordan algebra isomorphism $A$ between $\J$ and $\J'$ such
  that $A(\qk \oplus \CC x_1) = \qk \oplus \CC x'_1$ if and only if
  there exist an invertible map $P \in \End(\qk)$ and a nonzero
  $\lambda \in \CC$ such that $\lambda C' = P C P^{-1}$ and $P^* P C =
  C$, where $P^*$ is the adjoint map of $P$ with respect to $B$. In
  this case $A$ i-isomorphic then $P\in \OO(\qk)$.

\item Let $\J = \qk \oplusp (\CC x_1 \oplus \CC y_1)$ and $\J' = \qk
  \oplusp (\CC x'_1 \oplus \CC y'_1)$ be diagonalizable double
  extensions of $\qk$ by symmetric maps $C$ and $C'$
  respectively. Then $\J$ and $\J'$ are isomorphic if and only if they
  are i-isomorphic. In this case, $C$ and $C'$ have the same spectrum.
\end{enumerate}}

%

In Section 3, we introduce the notion of generalized double extension
but with a restricting condition for 2-step nilpotent pseudo-Euclidean
Jordan algebras. As a consequence, we obtain in this way the inductive
characterization of those algebras (Proposition \ref{prop3.11}):
\smallskip

THEOREM 3:

{\em Let $\J$ be a 2-step nilpotent pseudo-Euclidean Jordan
  algebra. If $\J$ is non-Abelian then it is obtained from an Abelian
  algebra by a sequence of generalized double extensions.}

To characterize (up to isomorphisms and i-isomorphisms) 2-step
nilpotent pseudo-Euclidean Jordan algebras we need to use the concept
of a $T^*$-extension in \cite{Bor97} as follows. Given a complex
vector space $\ak$ and a non-degenerate cyclic symmetric bilinear map
$\theta: \ak\times\ak \rightarrow \ak^*$, define on the vector space
$\J= \ak\oplus\ak^*$ the product \[
(x+f)(y+g) = \theta(x,y)\] then $\J$ is a 2-step nilpotent
pseudo-Euclidean Jordan algebra and it is called a {\em
  $T^*$-extension of $\ak$ by $\theta$} (or a $T^*$-extension,
simply). Moreover, we have the following result (Proposition
\ref{prop3.14}):

\smallskip

THEOREM 4:{\em

  Every reduced 2-step nilpotent pseudo-Euclidean Jordan algebra is
  i-isomorphic to some $T^*$-extension.}

Theorem 4 allows us to consider only isomorphic classes and
i-isomorphic classes of $T^*$-extensions to represent all 2-step
nilpotent pseudo-Euclidean Jordan algebras. An i-isomorphic and
isomorphic characterization of $T^*$-extensions is given by:
\smallskip

THEOREM 5:{\em

Let $\J_1$ and $\J_2$ be $T^*$-extensions of $\ak$ by $\theta_1$ and $\theta_2$ respectively. Then:
\begin{enumerate}
\item there exists a Jordan algebra isomorphism between $\J_1$ and
  $\J_2$ if and only if there exist an isomorphism $A_1$ of $\ak$ and
  an isomorphism $A_2$ of $\ak^*$ satisfying:
\[A_2(\theta_1(x,y)) = \theta_2(A_1(x),A_1(y)), \forall x,y\in\ak.\]

\item there exists a Jordan algebra i-isomorphism between $\J_1$ and
  $\J_2$ if and only if there exists an isomorphism $A_1$ of $\ak$
\[\theta_1(x,y) = \theta_2(A_1(x),A_1(y))\circ A_1, \forall x,y\in\ak.\]
\end{enumerate}}

As a consequence, the classification of i-isomorphic $T^*$-extensions
of $\ak$ is equivalent to the classification of symmetric 3-forms on
$\ak$. We detail it in the cases of $\dim(\ak) = 1$ and $2$.

In the last Section, we study Novikov algebras. These objects appear
in the study of the Hamiltonian condition of an operator in the formal
calculus of variations \cite{GD79} and in the classification of
Poisson brackets of hydrodynamic type
\cite{BN85}. 
A detailed classification of Novikov algebras up to dimension 3 can be
found in \cite{BM01}.

An associative algebra is both Lie-admissible and
Jordan-admissible. This is not true for Novikov algebras although they
are Lie-admissible. Therefore, it is natural to search a condition for
a Novikov algebra to become Jordan-admissible. The condition we give
here (weaker than associativity) is the following (Proposition
\ref{prop4.17}):

\smallskip

THEOREM 6:{\em

A Novikov algebra $\N$ is Jordan-admissible if it satisfies the condition
\[(x,x,x) = 0,\forall x\in\N.\] }
A corollary of Theorem 6 is that Novikov algebras are not
power-associative since there exist Novikov algebras not
Jordan-admissible.

Next, we consider symmetric Novikov algebras. A Novikov algebra $\N$
is called {\em symmetric} if it is endowed with a non-degenerate
associative symmetric bilinear form. In this case, $\N$ will be
associative, its sub-adjacent Lie algebra $\g(\N)$ is a quadratic
2-step nilpotent Lie algebra \cite{AB10} and the associated Jordan
algebra $\J(\N)$ is pseudo-Euclidean. Therefore, the study of
quadratic 2-step nilpotent Lie algebras (\cite{Ova07}, \cite{Duo10})
and pseudo-Euclidean Jordan algebras is closely related to symmetric
Novikov algebras.

By the results in \cite{ZC07} and \cite{AB10}, we have that every
symmetric Novikov algebra up to dimension 5 is commutative and a
non-commutative example is given in the case of dimension 6. This
algebra is 2-step nilpotent. In this paper, we show that every
symmetric non-commutative Novikov algebra of dimension 6 is 2-step
nilpotent.

As for quadratic Lie algebras and pseudo-Euclidean Jordan algebras, we
define the notion of a {\em reduced} symmetric Novikov algebra. Using
this notion, we obtain (Proposition \ref{prop4.28}):

\smallskip

THEOREM 7:

{\em Let $\N$ be a symmetric non-commutative Novikov algebra. If $\N$
  is reduced then
\[
3\leq \dim(\Ann(\N))\leq \dim(\N\N) \leq\dim(\N) -3.\]
}

In other words, we do not have $\N\N = \N$ in the non-commutative
case. Note that this may be true in the commutative case (see Example
\ref{ex4.14}). As a result, we obtain the following result for the
case of dimension 7 (Proposition \ref{prop4.30}):

THEOREM 8:

{\em Let $\N$ be a symmetric non-commutative Novikov algebra of dimension 7. If $\N$ is reduced then there are only two cases:
  \begin{enumerate}
    \item $\N$ is 3-step nilpotent and indecomposable.
    \item $\N$ is decomposable by $\N =\CC x\oplusp \N_6$, where $x^2=x$ and $\N_6$ is a symmetric non-commutative Novikov algebra of dimension 6.
\end{enumerate}}

Finally, we give an example for 3-step nilpotent symmetric Novikov
algebras of dimension 7. By the above theorem, it is indecomposable.

\begin{aknow}
  We heartily thank Didier Arnal for many discussions and suggestions
  for the improvement of this paper.

  This article is dedicated to our mentor Georges Pinczon (1948 --
  2010), an admirable mathematician, a very fine algebrist and above
  all, a very good friend.
\end{aknow}

\section{Pseudo-Euclidean Jordan algebras}

\begin{defn}
  A (non-associative) algebra $\J$ over $\CC$ is called a
  (commutative) {\em Jordan algebra} if its product is commutative and
  satisfies the following identity ({\em Jordan identity}):
  \begin{equation} \label{I} (xy)x^2 = x(yx^2), \forall \ x, y, z \in
    \J.
\end{equation}
\end{defn}

For instance, any commutative algebra with an associative product is a
Jordan algebra.

Given an algebra $A$, the {\em commutator} $[x,y] := xy-yx$, $\forall
\ x,y \in A$ measures the commutativity of $A$. Similarly the {\em
  associator} defined by \[ (x,y,z):= (xy)z - x(yz), \ \forall \ x, y,
z \in A.\] measures the associativity of $A$. In terms of associators,
the Jordan identity in a Jordan algebra $\J$ becomes
\begin{equation} \label{I'} (x, y, x^2) = 0, \forall \ x, y, z \in \J.
  \end{equation}

  An algebra $A$ is called a power-associative algebra if the
  subalgebra generated by any element $x \in A$ is associative (see
  \cite{Sch61} for more details). A Jordan algebra is an example of a
  power-associative algebra. A power-associative algebra $A$ is called
  {\em trace-admissible} if there exists a bilinear form $\tau$ on $A$
  that satisfies:
\begin{enumerate}
    \item $\tau(x,y) = \tau(y,x)$,
    \item $\tau(xy,z) = \tau(x,yz)$,
    \item $\tau(e,e) \neq 0$ for any idempotent $e$ of $A$,
    \item $\tau(x,y) = 0$ if $xy$ is nilpotent or $xy =0$.
\end{enumerate}

It is a well-known result that simple (commutative) Jordan algebras
are trace-admissible \cite{Alb49}. A similar fact is proved for any
{\em non-commutative} Jordan algebras of characteristic 0
\cite{Sch55}. Recall that non-commutative Jordan algebras are algebras
satisfying (\ref{I}) and the {\em flexible} condition $(xy)x = x(yx)$
(a weaker condition than commutativity).

A bilinear form $B$ on a Jordan algebra $\J$ is {\em associative}
if \[B(xy,z) = B(x, yz), \forall \ x,y,z\in \J.\] The following
definition is quite natural:

\begin{defn}
  Let $\J$ be a Jordan algebra equipped with an associative symmetric
  non-degenerate bilinear form $B$.  We say that the pair $(\J, B)$ is
  a {\em pseudo-Euclidean Jordan algebra} and $B$ is an {\em
    associative scalar product} on $\J$.
\end{defn}

Recall that a real finite-dimensional Jordan algebra $\J$ with a unit
element $e$ (that means, $xe=ex=x,\ \forall x\in\J$) is called {\em
  Euclidean} if there exists an associative inner product on
$\J$. 
This is equivalent to say that the associated trace form $\Tr(xy)$ is
positive definite, where $\Tr(x)$ is the sum of eigenvalues in the
spectral decomposition of $x \in \J$. To obtain a pseudo-Euclidean
Jordan algebra, we replace the base field $\RR$ by $\CC$ and the inner
product by a non-degenerate symmetric bilinear form (considered as a
generalized inner product) on $\J$ keeping its associativity.

\begin{lem}\label{lem1.3}
  Let $(\J,B)$ be a pseudo-Euclidean Jordan algebra and $I$ be a {\bf non-degenerate ideal} of $\J$, that is, the
  restriction $B|_{I \times I}$ is non-degenerate. Then $I^\bot$ is
  also an ideal of $\J$, $II^\bot = I^\bot I = \{0\}$ and $I\cap
  I^\bot = \{0\}$.
\end{lem}

\begin{proof}
  Let $x\in I^\bot, y\in\J$, one has $B(xy,I)= B(x,yI) = 0$
  then $xy\in I^\bot$ and $I^\bot$ is an ideal.

  If $x\in I^\bot$ such that $B(x,I^\bot) = 0$ then $x\in I$ and
  $B(x,I) =0$. Since $I$ is non-degenerate then $x=0$. That implies
  that $I^\bot$ is non-degenerate.

  Since $B(II^\bot,\J) =  B(I,I^\bot\J) = 0$ then $II^\bot = I^\bot I = \{0\}$.

  If $x\in I\cap I^\bot$ then $B(x,I) = 0$. Since $I$ non-degenerate,
  then $x=0$.
\end{proof}

By the proof of above Lemma, given a non-degenerate subspace $W$ of
$\J$ then $W^\bot$ is also non-degenerate and $\J=W\oplus W^\bot$. In
this case, we use the notation:
\[\J = W\oplusp W^\bot.\]

\begin{rem}
  A pseudo-Euclidean Jordan algebra does not necessarily have a unit
  element. However if that is the case, this unit element is certainly
  unique. A Jordan algebra with unit element is called a {\em unital}
  Jordan algebra. If $\J$ is not a unital Jordan algebra, we can
  extend $\J$ to a unital Jordan algebra $\overline{\J} = \CC e\oplus
  \J$ by the product
  \[(\lambda e + x)\star(\mu e +y ) = \lambda \mu e + \lambda y + \mu
  x + xy.\] More particularly, $e\star e =e$, $e\star x=x\star e = x$
  and $x\star y = xy$. In this case, we say $\overline{\J}$ the {\em
    unital extension} of $\J$.
\end{rem}
\begin{prop}
  If $(\J,B)$ is unital then there is a decomposition:
  \[\J = \J_1\oplusp\dots\oplusp\J_k,\]
  where $\J_i,\ i=1,\dots,k$ are unital and indecomposable ideals.
\end{prop}

\begin{proof}
  The assertion is obvious if $\J$ is indecomposable. Assume that $\J$
  is decomposable, that is, $\J = I\oplus I'$ with $I,\ I' \neq \{0\}$
  proper ideals of $\J$ such that $I$ is non-degenerate. By the above
  Lemma, $I' = I^\bot$ and we write $\J = I\oplusp I^\bot$. Assume
  that $\J$ has the unit element $e$. If $e\in I$ then for $x$ a
  nonzero element in $I^\bot$, we have $e x = x \in I^\bot$. This is a
  contradiction. This happens similarly if $e\in I^\bot$. Therefore,
  $e = e_1+e_2$ where $e_1\in I$ and $e_2\in I^\bot$ are nonzero
  vectors. For all $x\in I$, one has:
  \[x = ex = (e_1+e_2)x = e_1x = xe_1.\] It implies that $e_1$ is the
  unit element of $I$. Similarly, $e_2$ is also the unit element of
  $I^\bot$. Since the dimension of $\J$ is finite then by induction,
  one has the result.
\end{proof}

\begin{ex}\label{ex1.6}
  Let us recall an example in Chapter II of \cite{FK94}: consider
  $\qk$ a vector space over $\CC$ and $B:\qk\times\qk\rightarrow \CC$
  a symmetric bilinear form. Define the product below on the vector
  space $\J = \CC e\oplus \qk$:
\[(\lambda e + u)(\mu e + v) := (\lambda\mu + B(u,v))e + \lambda v +
\mu u,
\]
for all $\lambda,\mu\in\CC, u,v\in\qk$. In particular, $e^2=e$, $ue=eu
= u$ and $uv=B(u,v)e$. This product makes $\J$ a Jordan
algebra.

Now, we add the condition that $B$ is non-degenerate and define a
bilinear form $B_\J$ on $\J$ by:
\[B_\J(e,e) = 1, \ B_\J(e,\qk) = B_\J(\qk,e) = 0 \text{ and
}B_\J|_{\qk\times\qk} = B.
\]
Then $B_\J$ is associative and non-degenerate and $\J$ becomes a
pseudo-Euclidean Jordan algebra with unit element $e$.

\end{ex}

\begin{ex}
  Let us slightly change Example \ref{ex1.6} by setting \[\J' := \CC
  e\oplus \qk\oplus \CC f.\] Define the product of $\J'$ as follows:
  \[e^2=e,\ ue=eu = u, \ ef=fe=f,\ uv=B(u,v)f \text{ and } uf = fu
  =ff=0,\] for all $u,v\in\qk$. It is easy to see that $\J'$ is the unital extension of the
  Jordan algebra $\J = \qk\oplus \CC f$, where the product on $\J$ is
  defined by:
  \[uv=B(u,v)f, \ uf = fu= 0, \forall u,v\in \qk.\] Moreover, $\J'$ is
  a pseudo-Euclidean Jordan algebra with the bilinear form $B_{\J'}$
  defined by:
  \[B_{\J'}\left( \lambda e + u + \lambda' f,\mu e + v + \mu' f\right)
  = \lambda \mu'+\lambda' \mu + B(u,v).\] We will meet this algebra
  again in the next Section.
\end{ex}

Recall the definition of a representation of a Jordan algebra:

\begin{defn}
  A {\em Jacobson representation} (or simply, a {\em representation})
  of a Jordan algebra $\J$ on a vector space $V$ is a linear map $\J
  \to \End(V)$, $x \mapsto S_x$ satisfying for all $x$, $y$, $z \in \J$,

\begin{enumerate}
\item $[S_x, S_{yz}] + [S_y, S_{zx}] + [S_z, S_{xy}] = 0$,

\item $S_x S_y S_z + S_z S_y S_x + S_{(xz)y} = S_x S_{yz} + S_y S_{zx}
  + S_z S_{xy}.$

\end{enumerate}

\end{defn}

\begin{rem}
  An equivalent definition of a representation of $\J$ can be found
  for instance in \cite{BB08}, as a necessary and sufficient condition
  for the vector space $\J_1 = \J \oplus V$ equipped with the
  product: \[(x+u)(y+v) = xy + S_x(v)+S_y(u),\ \forall
  x,y\in\J,u,v\in V \] to be a Jordan algebra. In this case,
  Jacobson's definition is different from the usual definition of
  representation, that is, as a homomorphism from $\J$ into the Jordan
  algebra of linear maps.
\end{rem}

For $x \in \J$, let $R_x \in \End(\J)$ be the endomorphism of $\J$
defined by: \[ R_{x}(y) = xy =yx, \forall \ y \in \J.\] Then the
Jordan identity is equivalent to $[R_{x}, R_{x^2}] = 0, \forall \ x
\in \J$ where $[\cdot, \cdot]$ denotes the Lie bracket on
$\End(\J)$. The linear maps \[ R: \J \rightarrow \End(\J) \; \text{
  with } R(x) := R_{x}\] \[ \; \text{ and } R^*: \J
\rightarrow \End(\J^*) \; \text{ with } R^*(x)(f) = f\circ R_x,
\forall \ x \in \J, f \in \J^*,\] are called respectively the {\em
  adjoint representation} and the {\em coadjoint representation} of
$\J$. It is easy to check that they are indeed representations of
$\J$. Recall that there exists a natural non-degenerate bilinear from
$\langle \cdot, \cdot \rangle$ on $\J\oplus\J^*$ defined by $\langle
x,f \rangle := f(x),\ \forall x\in\J,\ f\in\J^*$. For all
$x,y\in\J,f\in\J^*$, one has:
\[ f(xy) = \langle xy,f \rangle = \langle R_x(y),f \rangle = \langle
y,R_x^*(f) \rangle.\] That means that $R_x^*$ is the adjoint map of
$R_x$ with respect to the bilinear form $\langle \cdot, \cdot
\rangle$.

The following proposition gives a characterization of pseudo-Euclidean
Jordan algebras. A proof can be found in \cite{BB08}, Proposition 2.1
or \cite{Bor97}, Proposition 2.4.

\begin{prop}
  Let $\J$ be a Jordan algebra. Then $\J$ is pseudo-Euclidean if, and
  only if, its adjoint representation and coadjoint representation are
  equivalent.
\end{prop}

We will need some special subspaces of an arbitrary algebra $\J$:

\begin{defn}\label{1.11}
Let $\J$ be an algebra.
\begin{enumerate}
\item The subspace \[(\J, \J, \J) := \spa\{(x,y,z) \mid x,y, z\in
  \J\}\] is the {\em associator} of $\J$.

\item The subspaces \[\LAnn(\J):= \{x\in \J \mid x\J = 0\}, \] \[
  \RAnn(\J):= \{x\in \J \mid \J x = 0\} \text{ and } \] \[\Ann(\J):=
  \{x\in \J \mid x\J = \J x = 0\}\] are respectively the {\em
    left-annulator}, the {\em right-annulator} and the {\em annulator}
  of $\J$. Certainly, if $\J$ is commutative then these subspaces
  coincide.

\item The subspace \[N(\J):= \{x\in \J \mid (x,y,z) = (y,x,z) =
  (y,z,x) = 0, \forall y,z \in \J \}\] is the {\em nucleus} of $\J$.

\end{enumerate}
\end{defn}

The proof of the Proposition below is straightforward and we omit it.

\begin{prop} If $(\J, B)$ is a pseudo-Euclidean Jordan algebra then
  \begin{enumerate}
  \item the nucleus $N(\J)$ coincide with the {\bf center} $\Zb(\J)$
    of $\J$ where $\Zb(\J) = \{ x \in N(\J) \mid xy = yx, \forall y
    \in \J\}$, that is, the set of all elements $x$ that commute and
    associate with all elements of $\J$. Therefore \[ N(\J)= \Zb(\J) =
    \{x\in \J \mid (x,y,z) =0, \forall y,z\in \J\}.\]
    \item $\Zb(\J)^\perp = (\J, \J, \J)$.
    \item $\left( \Ann(\J) \right)^\perp = \J^2$.
 \end{enumerate}
\end{prop}


Just as in \cite{DPU10} where we have defined reduced quadratic Lie
algebras, we can define here:

\begin{defn}
A pseudo-Euclidean Jordan algebra $(\J,B)$ is {\em reduced} if
\begin{enumerate}
    \item $\J\neq \{0\}$,
    \item $\Ann(\J)$ is totally isotropic, that means $B(x,y)=0$ for all $x,y\in\Ann(\J)$.
\end{enumerate}
\end{defn}

\begin{prop}\label{prop-red}
  Let $\J$ be non-Abelian pseudo-Euclidean Jordan algebra. Then $\J =
  \zk\oplusp\lk$, where $\zk\subset \Ann(\J)$ and $\lk$ is reduced.
\end{prop}

\begin{proof}
  The proof is completely similar to Proposition 6.7 in \cite{PU07}.
  Let $\zk_0 = \Ann(\J)\cap \J^2$, $\zk$ is a complementary subspace
  of $\zk_0$ in $\Ann(\J)$ and $\lk = \zk^\bot$. If $x$ is an element
  in $\zk$ such that $B(x,\zk) = 0$ then $B(x,\J^2) = 0$ since
  $\Ann(\J) = (\J^2)^\perp$. As a consequence, $B(x,\zk_0) = 0$ and
  therefore $B(x,\Ann(\J)) = 0$. That implies $x\in\J^2$. Hence, $x=0$
  and the restriction of $B$ to $\zk$ is non-degenerate. Moreover,
  $\zk$ is an ideal then by Lemma \ref{lem1.3}, the restriction of
  $B$ to $\lk$ is also a non-degenerate and that $\zk \cap \lk =
  \{0\}$.

  Since $\J$ is non-Abelian then $\lk$ is non-Abelian and $\lk^2 =
  \J^2$. Moreover, $\zk_0=\Ann(\lk)$ and the result follows.
\end{proof}

Next, we will define some extensions of a Jordan algebra and introduce
the notion of a {\em double extension} of a pseudo-Euclidean Jordan
algebra \cite{BB08}.








\begin{defn}\label{4.1.7}
  Let $\J_1$ and $\J_2$ be Jordan algebras and $\pi: \J_{1}
  \rightarrow \End(\J_{2})$ be a representation of $\J_1$ on
  $\J_2$. We call $\pi$ an {\em admissible representation} if it
  satisfies the following conditions:
\begin{enumerate}
\item $\pi(x^2)(yy') + 2(\pi(x)y')(\pi(x)y) +(\pi(x)y')y^2 +
  2(yy')(\pi(x)y)\\ = 2\pi(x)(y'(\pi(x)y)) + \pi(x)(y'y^2) +
  (\pi(x^2)y')y + 2(y'(\pi(x)y))y,$

\item $(\pi(x)y)y^2 = (\pi(x)y^2)y$,

\item $\pi(xx')y^2 + 2(\pi(x')y)(\pi(x)y) =
       \pi(x)\pi(x')y^2 + 2(\pi(x')\pi(x)y)y$,
\end{enumerate}

for all $x,x'\in \J_{1}, y, y' \in \J_{2}$. In this case, the vector
space $\J = \J_1\oplus\J_2$ with the product defined
by: \[(x+y)(x'+y') = xx'+\pi(x)y'+\pi(x')y+yy', \ \forall x,x'\in
\J_1, y,y'\in \J_2 \] becomes a Jordan algebra.
\end{defn}

\begin{defn}
  Let $(\J, B)$ be a pseudo-Euclidean Jordan algebra and $C$ be an
  endomorphism of $\J$. We say that $C$ is {\em symmetric} if
\[
B(C(x),y) = B(x, C(y)), \forall x,y \in \J.
\]

Denote by $\End_{s}(\J)$ the space of symmetric endomorphisms of $\J$.

\end{defn}

The definition below was introduced in \cite{BB08}, Theorem 3.8.

\begin{defn} \label{defde}

  Let $(\J_{1}, B_{1})$ be a pseudo-Euclidean Jordan algebra, $\J_{2}$
  be an arbitrary Jordan algebra and $\pi: \J_{2}
  \rightarrow \End_{s}(\J_{1})$ be an admissible
  representation. Define a symmetric bilinear map $\varphi:
  \J_{1}\times \J_{1} \rightarrow \J_{2}^*$ by: $\varphi(y,y')(x) =
  B_{1}(\pi(x)y,y'), \forall x\in \J_{2}, y,y'\in \J_{1}$. Consider
  the vector space \[\Jb = \J_{2}\oplus \J_{1}\oplus \J_{2}^*\]
  endowed with the product:
\[
(x+y+f)(x'+y'+f') = xx' + yy' +\pi(x)y' + \pi(x')y + f'\circ R_{x} +
f\circ R_{x'} + \varphi(y,y')
\]
for all $x, x' \in \J_{2}, y, y'\in \J_{1}, f, f' \in \J_{2}^*$. Then
$\Jb$ is a Jordan algebra. Moreover, define a bilinear form $B$ on
$\Jb$ by:
\[
B(x+y+f,x'+y'+f') = B_{1}(y,y') + f(x') + f'(x), \forall x, x' \in
\J_{2}, y, y'\in \J_{1}, f, f' \in \J_{2}^*.
\] Then $\Jb$ is a {\bf pseudo-Euclidean Jordan algebra}. The Jordan
algebra $(\Jb, B)$ is called the {\em double extension} of $\J_{1}$ by
$\J_{2}$ by means of $\pi$.
\end{defn}

\begin{rem}\label{rem1.18}
  If $\gamma$ is an associative bilinear form (not necessarily
  non-degenerate) on $\J_2$ then $\Jb$ is again pseudo-Euclidean
  thanks to the bilinear form
  \[B_\gamma(x+y+f,x'+y'+f') = \gamma(x,x')+ B_{1}(y,y') + f(x') +
  f'(x)\] for all $x$, $x' \in \J_{2}$, $y$, $y'\in \J_{1}$, $f$, $f'
  \in \J_{2}^*$.
\end{rem}

\section{Jordanian double extension of a quadratic vector space}

Let $\CC c$ be a one-dimensional Jordan algebra. If $c^2 \neq 0$ then
$c^2 = \lambda c$ for some nonzero $\lambda\in \CC$. Replace $c:=
\frac{1}{\lambda}c$, we obtain $c^2= c$. Therefore, there exist only
two one-dimensional Jordan algebras: one Abelian and one simple. Next,
we will study 
double extensions of a quadratic vector space by these algebras.

Let us start with $(\qk, B_{\qk})$ a {\bf quadratic vector space},
that is, $B_{\qk}$ is a non-degenerate symmetric bilinear form on the
vector space $\qk$. We consider $(\tk = \spa \{x_{1}, y_{1}\},
B_{\tk})$ a 2-dimensional quadratic vector space with the bilinear form $B_{\tk}$
defined by
\[
B_\tk(x_1, x_1) = B_\tk(y_1, y_1) = 0, \ B_\tk(x_1,
  y_1) = 1.
\]
Let $C: \qk \rightarrow \qk$ be a nonzero symmetric map and consider
the vector space \[\J = \qk \oplusp \tk\] equipped with a product
defined by
\begin{eqnarray*}
  &(x + \lambda x_1 + \mu y_1)&( y + \lambda' x_1 + \mu' y_1)  := \\ &&\mu C(y) + \mu'C(x) + B_\qk( C(x), y) x_1 + \e \left( \left( \lambda\mu' + \lambda'\mu \right) x_1 + \mu \mu' y_1\right),
  \end{eqnarray*}
  $\e \in \{0,1 \}$, for all $x, y \in \qk, \lambda, \mu, \lambda',
  \mu' \in \CC$.

  \begin{prop}\label{prop2.1} Keep the notation just above.

  \begin{enumerate}

  \item Assume $\e = 0$. Then $\J$ is a Jordan algebra if, and only
    if, $C^3 = 0$. In this case, we call $\J$ a {\bf nilpotent double
      extension} of $\qk$ by $C$.

  \item Assume $\e = 1$. Then $\J$ is a Jordan algebra if, and only
    if, $3C^2 = 2C^3 + C$.  Moreover, $\J$ is pseudo-Euclidean with
    the bilinear form $B = B_\qk + B_\tk$. In this case, we call $\J$ a
    {\bf diagonalizable double extension} of $\qk$ by $C$.
  \end{enumerate}

  \end{prop}

  \begin{proof}\hfill
    \begin{enumerate}

\item Let $x$, $y \in \qk$, $\lambda$, $\mu$, $\lambda'$, $\mu' \in
  \CC$. One has
\[
((x + \lambda x_1 + \mu y_1)( y + \lambda' x_1 + \mu' y_1))(x +
\lambda x_1 + \mu y_1)^2 = 2\mu B_\qk(C^2 (\mu y +\mu' x), C(x)) x_{1}
\]
and \[(x + \lambda x_1 + \mu y_1)(( y + \lambda' x_1 + \mu' y_1)(x +
\lambda x_1 + \mu y_1)^2) = 2\mu^2\mu' C^3(x)\] \[+
2\mu\mu'B_\qk(C(x),C^2(x))x_{1}.\] Therefore, $\J$ is a Jordan algebra
if and only if $C^3 = 0$.

\item The result is achieved by checking directly the equality
  (\ref{I}) for $\J$.

\end{enumerate}

\end{proof}

\subsection{Nilpotent double extensions} \hfill

Consider $\J_1:=\qk$ an Abelian algebra, $\J_2: = \CC y_1$ the
nilpotent one-dimensional Jordan algebra, $\pi(y_1):=C$ and identify
$\J_2^*$ with $\CC x_1$. Then by Definition \ref{defde}, $\J =
\J_{2}\oplus \J_{1}\oplus \J_{2}^*$ is a pseudo-Euclidean Jordan
algebra with a bilinear form $B$ given by $B := B_\qk + B_\tk$.  In
this case, $C$ obviously satisfies the condition $C^3=0$.

An immediate corollary of the definition is:

\begin{cor}\label{cor2.2}
  If $\J= \qk \oplusp (\CC x_1 \oplus \CC y_1)$ is the nilpotent
  double extension of $\qk$ by $C$ then \[ y_1x = C(x),
  xy=B(C(x),y)x_1 \text{ and } y_1y_1 = x_1 \J = 0, \forall x\in \qk.
  \] As a consequence, $\J^2 = \im(C)\oplus\CC x_1$ and $Ann(\J) =
  \ker(C)\oplus\CC x_1$.
\end{cor}

\begin{rem}\label{rem2.3}
  In this case, $\J$ is $k$-step nilpotent, $k\leq 3$ since
  $R_x^k(\J)\subset \im(C^k)\oplus\CC x_1$.
\end{rem}

\begin{defn}
  Let $(V,B)$ and $(V',B')$ be two quadratic vector spaces. An {\em
    isometry} is a bijective map $A: V \rightarrow V'$ that satisfies
\[B'(A(v),A(w)) = B(v,w),\ \forall u,v\in V.
\]
The group of isometries of $V$ is denoted by $\OO(V,B)$ (or simply
$\OO(V)$). In the case $(\J,B)$ and $(\J',B')$ are pseudo-Euclidean
Jordan algebras, if there exists a Jordan algebra isomorphism $A$
between $\J$ and $\J'$ such that it is also an isometry then we say
$\J$, $\J'$ are {\em i-isomorphic} and $A$ is an {\em i-isomorphism}.
\end{defn}

\begin{prop} \label{prop2.3}
  Let $(\qk, B)$ be a quadratic vector space. Let $\J = \qk \oplusp
  (\CC x_1 \oplus \CC y_1)$ and $\J' = \qk \oplusp (\CC x'_1 \oplus
  \CC y'_1)$ be nilpotent double extensions of $\qk$, by symmetric
  maps $C$ and $C'$ respectively. Then:
\begin{enumerate}
\item there exists a Jordan algebra isomorphism $A$ between $\J$ and
  $\J'$ such that $A(\qk \oplus \CC x_1) = \qk \oplus \CC x'_1$ if, and
  only if, there exists an invertible map $P \in \End(\qk)$ and a
  nonzero $\lambda \in \CC$ such that $\lambda C' = P C P^{-1}$ and
  $P^* P C = C$, where $P^*$ is the adjoint map of $P$ with respect to
  $B$.

\item there exists a Jordan algebra i-isomorphism $A$ between $\J$ and $\J'$ such that $A(\qk \oplus \CC x_1) = \qk
  \oplus \CC x'_1$ if, and only if, there exists a nonzero $\lambda \in
  \CC$ such that $C$ and $\lambda C'$ are conjugate by an isometry
  $P\in \OO(\qk)$.
\end{enumerate}
\end{prop}

\begin{proof}\hfill
\begin{enumerate}
\item Assume $A: \J \rightarrow \J'$ be an isomorphism such that
  $A(\qk \oplus \CC x_1) = \qk \oplus \CC x'_1$. Since $x_1 \in \J^2$,
  then there exist $x,y \in \J$ such that $xy = x_1$ (by Proposition
  \ref{prop2.1}). Therefore $A(x_1) = A(x)A(y) \in (\qk \oplus \CC
  x'_1)(\qk \oplus \CC x'_1) = \CC x'_1$. That means $A(x_1) = \mu
  x'_1$ for some nonzero $\mu \in \CC$. Write $A|_\qk = P + \beta
  \otimes x'_1$ with $P \in \End(\qk)$ and $\beta \in \qk^*$. If $x\in
  \ker(P)$ then $ A \left( x - \dfrac{1}{\mu} \beta(x) x_1 \right) =
  0,$ so $x = 0$ and therefore, $P$ is invertible.  For all $x,y \in
  \qk$, one has \[ \mu B(C(x),y)x'_1 = A(xy)=A(x)A(y) =
  B(C'(P(x)),P(y))x'_1. \] So we obtain $P^*C'P = \mu C$.  Assume
  $A(y_1) = y + \delta x'_1 + \lambda y'_1 $, with $y\in \qk$. For all
  $x\in\qk$, one has \[ P(C(x)) + \beta (C(x))x'_1 = A(y_1x) =
  A(y_1)A(x) = \lambda C'(P(x)) + B(C'(y), P(x))x'_1. \] Therefore,
  $\lambda C' = PCP^{-1}$. Combine with $P^*C'P = \mu C$ to get $P^*PC
  = \lambda\mu C$. Replace $P$ by $\dfrac{1}{ (\mu \lambda)^{\frac12}}
  P$ to obtain $\lambda C' = PCP^{-1}$ and $P^*PC = C$.

Conversely, define $A: \J \rightarrow \J'$ by $A(y_1) = \lambda y'_1$,
$A(x) = P(x), \forall x\in \qk$ and $A(x_1) = \frac{1}{\lambda}x'_1$
then it is easy to check $A$ is an isomorphism.

\item If $A: \J \rightarrow \J'$ is an i-isomorphic then the
  isomorphism $P$ in the proof of (1) is also an isometry. Hence $P
  \in \OO(\qk)$. Conversely, define $A$ as in (1) then it is obvious
  that $A$ is an i-isomorphism.
\end{enumerate}

\end{proof}

\begin{prop}
  Let $(\qk, B)$ be a quadratic vector space, $\J = \qk \oplusp (\CC
  x_1 \oplus \CC y_1)$, $\J' = \qk \oplusp (\CC x'_1 \oplus \CC y'_1)$
  be nilpotent double extensions of $\qk$, by symmetric maps $C$ and
  $C'$ respectively. Assume that $\rank(C')\geq 3$. Let $A$ be an
  isomorphism between $\J$ and $\J'$. Then $A(\qk \oplus \CC x_1) =
  \qk \oplus \CC x'_1$.
\end{prop}

\begin{proof}
  We assume that there is $x\in \qk$ such that $A(x) = y + \beta x'_1
  + \gamma y'_1$, where $y\in \qk, \beta, \gamma \in \CC, \gamma \neq
  0$. Then for all $q \in \qk$ and $\lambda \in \CC$, we have
\[
A(x)(q + \lambda x'_1) = \gamma C'(q) + B(C'(y),q )x'_1.\] Therefore,
$\dim({A(x)(\qk \oplus \CC x'_1)}) \geq 3$. But $A$ is an isomorphism,
hence \[ A(x)(\qk \oplus \CC x'_1) \subset A(xA^{-1}(\qk \oplus \CC
x'_1)) \subset A(x(\qk \oplus \CC x_1 \oplus \CC y_1)) \subset A(\CC
C(x) \oplus \CC x_1). \]
This is a contradiction. Hence $A(\qk \oplus \CC x_1) = \qk \oplus \CC x'_1$.
\end{proof}

\subsection{Diagonalizable double extensions} \hfill

\begin{lem}\label{cor2.7}
  Let $\J = \qk \oplusp (\CC x_1 \oplus \CC y_1)$ be the diagonalizable double
  extension of $\qk$ by $C$. Then
\[
y_1 y_1 = y_1, y_1 x_1 = x_1, y_1 x = C(x), xy = B(C(x),y)x_1 \text{ and } x_1 x = x_1 x_1 = 0, \forall x\in \qk.
\]
\end{lem}

\medskip

Note that $x_1 \notin \Ann(\J)$. Let $x\in \qk$. Then $ x\in \Ann(\J)$
if and only if $x\in \ker(C)$. Moreover, $\J^2 = \im(C)\oplus (\CC x_1
\oplus \CC y_1)$. Therefore $\J$ is reduced if, and only if,
$\ker(C)\subset \im(C)$.

Let $x\in \im(C)$. Then there exists $y\in \qk$ such that
$x=C(y)$. Since $3C^2 = 2C^3 + C$, one has $3C(x)-2C^2(x) =
x$. Therefore, if $\J$ is reduced then $\ker(C) = \{0\}$ and $C$ is
invertible. That implies that $3C - 2C^2 = \Id$ and we have the
following proposition:

\begin{prop}\label{prop2.8}
  Let $(\qk, B)$ be a quadratic vector space. Let $\J = \qk \oplusp
  (\CC x_1 \oplus \CC y_1)$ and $\J' = \qk \oplusp (\CC x'_1 \oplus
  \CC y'_1)$ be diagonalizable double extensions of $\qk$, by invertible maps
  $C$ and $C'$ respectively. Then there exists a Jordan algebra
  isomorphism $A$ between $\J$ and $\J'$ if and only if there exists
  an isometry $P$ such that $C' = P C P^{-1}$. In this case, $\J$ and
  $\J'$ are also i-isomorphic.
\end{prop}

\begin{proof}
  Assume $\J$ and $\J'$ isomorphic by $A$. Firstly, we will show that
  $A(\qk \oplus \CC x_1) = \qk \oplus \CC x'_1$. Indeed, if $A(x_1) =
  y + \beta x'_1 + \gamma y'_1$, where $y \in \qk, \beta, \gamma \in
  \CC$, then
\[
0 = A(x_1 x_1) = A(x_1)A(x_1) = 2\gamma C'(y) +(2\beta \gamma + B(C'(y),y))x'_1 + \gamma^2 y'_1.
\]
Therefore, $\gamma = 0$. Similarly, if there exists $x\in\qk$ such
that $A(x) = z + \alpha x'_1 + \delta y'_1$, where $z \in \qk, \alpha,
\delta \in \CC$. Then
\[
B(C(x),x)A(x_1) = A(x x) = A(x)A(x) = 2\delta C'(y) +(2\alpha \delta +
B(C'(z),z))x'_1 + \delta^2 y'_1.
\]
That implies $\delta = 0$ and $A(\qk \oplus \CC x_1) = \qk \oplus \CC
x'_1$.

The rest of the proof follows exactly the proof of Proposition
\ref{prop2.3}, one has $A(x_1) = \mu x'_1$ for some nonzero $\mu \in
\CC$ and there is an isomorphism $P$ of $\qk$ such that $A|_\qk = P
+\beta\otimes x'_1$, where $\beta \in \qk^*$. Similarly as in the
proof of Proposition \ref{prop2.3}, one also has $P^* C'P = \mu C$,
where $P^*$ is the adjoint map of $P$ with respect to $B$. Assume
$A(y_1) = \lambda y'_1 + y + \delta x_1$. Since $A(y_1)A(y_1)=A(y_1)$,
one has $\lambda = 1$ and therefore $C' = PCP^{-1}$. Replace $P:=
\dfrac{1}{ (\mu)^{\frac12}} P$ to get $P^* P C = C$. However, since
$C$ is invertible then $P^* P = \Id$. That means that $P$ is an
isometry of $\qk$.

Conversely, define $A: \J \rightarrow \J'$ by $A(x_1) = x'_1$, $A(y_1)
= y'_1$ and $A(x) = P(x), \forall x\in \qk$ then A is an
i-isomorphism.
\end{proof}

An invertible symmetric endomorphism of $\qk$ satisfying $3C - 2C^2 = \Id$ is
diagonalizable by an orthogonal basis of eigenvectors with eigenvalues
1 and $\frac{1}{2}$ (see Appendix). Therefore, we have the following
corollary:

\begin{cor}\label{cor2.9}
  Let $(\qk, B)$ be a quadratic vector space. Let $\J = \qk \oplusp
  (\CC x_1 \oplus \CC y_1)$ and $\J' = \qk \oplusp (\CC x'_1 \oplus
  \CC y'_1)$ be diagonalizable double extensions of $\qk$, by invertible maps
  $C$ and $C'$ respectively. Then $\J$ and $\J'$ are isomorphic if ,
  and only if, $C$ and $C'$ have same spectrum.
\end{cor}

\begin{ex}\label{2.9}
  Let $\CC x$ be one-dimensional Abelian algebra, $\J = \CC x \oplusp
  (\CC x_1 \oplus \CC y_1)$ and $\J' = \CC x\oplusp (\CC x'_1 \oplus
  \CC y'_1)$ be diagonalizable double extensions of $\CC x$ by $C = \Id$ and
  $C' = \frac{1}{2}\Id$. In particular, the product on $\J$ and $\J'$
  are defined by:
\[ y_1^2=y_1, y_1x=x, y_1x_1 = x_1, x^2 = x_1;
\]
\[(y'_1)^2=y'_1, y'_1x=\frac{1}{2}x, y_1x_1 = x_1, x^2 = \frac{1}{2}x_1.
\]
Then $\J$ and $\J'$ are not isomorphic. Moreover, $\J'$ has no unit
element.
\end{ex}

\begin{rem}
  The i-isomorphic and isomorphic notions are not coincident in general. For
  example, the Jordan algebras $\J =\CC e$ with $e^2=e$, $B(e,e)=1$
  and $\J' = \CC e'$ with $e'e' =e'$, $B(e',e')=a\neq 1$ are
  isomorphic but not i-isomorphic.
\end{rem}

\section{Pseudo-Euclidean 2-step nilpotent Jordan algebras}

Quadratic 2-step nilpotent Lie algebras are characterized up to
isometric isomorphisms and up to isomorphisms in \cite{Ova07}. There is
a similar natural property in the case of pseudo-Euclidean 2-step
nilpotent Jordan algebras.

\subsection{2-step nilpotent Jordan algebras} \hfill

Let us redefine 2-step nilpotent Jordan algebras in a more convenient
way:

\begin{defn}\label{defn3.1}
  An algebra $\J$ over $\CC$ with a product $(x, y) \mapsto xy$ is
  called {\em 2-step nilpotent Jordan algebra} if it satisfies $xy =
  yx$ and $(xy)z = 0$ for all $x, y, z \in \J$. Sometimes, we use {\bf
    2SN-Jordan Algebra} as an abbreviation.
\end{defn}

The method of double extension is a fundamental tool used in
describing algebras that are endowed with an associative
non-degenerate bilinear form. This method is based on two principal
notions: central extension and semi-direct sum of two algebras. In the
next part, we will recall some definitions given in Section 3 of
\cite{BB08} but with a restricting condition for pseudo-Euclidean
2-step nilpotent Jordan algebras.


\begin{prop}\label{prop-un}
  Let $\J$ be a 2SN-Jordan algebra, $V$ be a vector space, $\varphi:
  \J \times\J \rightarrow V$ be a bilinear map and $\pi: \J
  \rightarrow \End(V)$ be a representation. Let \[ \Jb =\J\oplus V\]
  equipped with the following product:
  \[(x+u)(y+v) = xy + \pi(x)(v) + \pi(y)(u)+ \varphi(x,y), \forall
  x,y\in \J, u,v\in V.\] Then $\Jb$ is a 2SN-Jordan algebra if and
  only if for all $x,y,z\in \J$:

\begin{enumerate}

\item $\varphi$ is symmetric and $\varphi(xy,z) + \pi(z)(\varphi(x,y))= 0$,

\item $\pi(xy) = \pi(x)\pi(y) = 0$.

\end{enumerate}
\end{prop}

\begin{defn}\label{defn3.2}
  If $\pi$ is the trivial representation in Proposition \ref{prop-un},
  the Jordan algebra $\Jb$ is called the {\em 2SN-central extension}
  of $\J$ by $V$ (by means of $\varphi$).
\end{defn}

Remark that in a 2SN-central extension $\Jb$, the annulator
$\Ann(\Jb)$ contains the vector space $V$.

\begin{prop}
  Let $\J$ be a 2SN-Jordan algebra. Then $\J$ is a 2SN-central
  extension of an Abelian algebra.
\end{prop}

\begin{proof}
  Set $\hk:=\J/\J^2$ and $V:=\J^2$. Define $\varphi:
  \hk\times\hk\rightarrow V$ by $\varphi(p(x),p(y)) = xy, \forall
  x,y\in \J$, where $p:\J\rightarrow \hk$ is the canonical
  projection. Then $\hk$ is an Abelian algebra and $\J\cong \hk\oplus
  V$ is the 2SN-central extension of $\hk$ by means of $\varphi$.
\end{proof}

\begin{rem}
  It is easy to see that if $\J$ is a 2SN-Jordan algebra, then the
  coadjoint representation $R^*$ of $\J$ satisfies the condition on
  $\pi$ in Proposition \ref{prop-un} (2). For a trivial $\varphi$, we
  conclude that $\J\oplus \J^*$ is also a 2SN-Jordan algebra with
  respect to the coadjoint representation.
\end{rem}

\begin{defn}\label{defn3.6}
  Let $\J$ be a 2SN-Jordan algebra, $V$ and $W$ be two vector
  spaces. Let $\pi: \J \rightarrow \End(V)$ and $\rho: \J
  \rightarrow \End(W)$ be representations of $\J$. The {\em direct
    sum} $\pi\oplus \rho: \J \rightarrow \End(V\oplus W)$ of $\pi$ and
  $\rho$ is defined by
  \[ (\pi\oplus \rho)(x)(v+w) = \pi(x)(v) + \rho(x)(w), \forall x\in
  \J, v\in V, w\in W. \]
\end{defn}

\begin{prop}\label{prop3.7}
  Let $\J_1$ and $\J_2$ be 2SN-Jordan algebras and $\pi: \J_{1}
  \rightarrow \End(\J_{2})$ be a linear map. Let \[ \J = \J_{1}\oplus
  \J_{2}.\] Define the following product on $\J$:
\[
(x+y)(x'+y') = xx' + \pi(x)(y')+ \pi(x')(y) + yy', \forall x, x' \in
\J_{1}, y, y' \in \J_{2}.
\]
Then $\J$ is a 2SN-Jordan algebra if and only if $\pi$ satisfies:
\begin{enumerate}
    \item $\pi(xx') = \pi(x)\pi(x') = 0,$
    \item $\pi(x)(yy') = (\pi(x)(y))y' = 0$,
\end{enumerate}
for all $x, x' \in \J_{1}, y, y' \in \J_{2}$.

In this case, $\pi$ satisfies the conditions of Definition \ref{4.1.7}, it is called a {\bf 2SN-admissible
  representation} of $\J_1$ in $\J_2$ and we say that $\J$ is the {\bf
  semi-direct sum} of $\J_2$ by $\J_1$ by means of $\pi$.
\end{prop}

\begin{proof}
For all $x, x', x'' \in \J_{1}, y, y', y'' \in \J_{2}$, one has:
\begin{eqnarray*}
((x+y)(x'+y'))(x''+y'') = \pi(xx')(y'') + \pi(x'')(\pi(x)(y') +
\pi(x')(y) + yy') \\ + (\pi(x)(y') + \pi(x')(y))y''.
\end{eqnarray*}

Therefore, $\J$ is 2-step nilpotent if, and only if, $\pi(xx')$,
$\pi(x)\pi(x')$, $\pi(x)(yy')$ and $(\pi(x)y)y'$ are zero, $\forall x,
x'\in \J_{1}, y, y'\in \J_{2}$.
\end{proof}

\begin{rem}\hfill\label{rem3.8}

\begin{enumerate}

\item The adjoint representation of a 2SN-Jordan algebra is an
  2SN-admissible representation.

\item Consider the particular case of $\J_{1} = \CC c$ a
  one-dimensional algebra. If $\J_{1}$ is 2-step nilpotent then $c^2
  =0$. Let $D:= \pi(c)\in \End(\J_{2})$. The vector space $\J = \CC
  c \oplus \J_2$ with the product:
  \[
  (\alpha c +x)(\alpha' c +x') = \alpha D(x') + \alpha' D(x) + xx',
  \forall x,x' \in \J_{2}, \alpha, \alpha' \in \CC.
\]
is a 2-step nilpotent if and only if $D^2 = 0$, $D(xx') = D(x)x' = 0,
\forall x,x' \in \J_{2}$.
\item Let us slightly change (2) by fixing $x_0\in\J_2$ and setting
  the product on $\J = \CC c \oplus \J_2$ as follows:
  \[(\alpha c +x)(\alpha' c +x') = \alpha D(x') + \alpha' D(x) + xx' +
  \alpha\alpha'x_0,
  \]
  for all $x,x' \in \J_{2}, \alpha, \alpha' \in \CC$. Then $\J$ is a
  2SN-Jordan algebra if, and only if:
  \[D^2(x) = D(xx') = D(x)x' = D(x_0) = x_0x = 0,\ \forall
  x,x'\in\J_2.
  \]
  In this case, we say $(D,x_0)$ a {\em 2SN-admissible pair} of $\J_2$.
\end{enumerate}
\end{rem}

Next, we see how to obtain a 2SN-Jordan algebra from a
pseudo-Euclidean one.

\begin{prop}
  Let $(\J, B)$ be a 2-step nilpotent pseudo-Euclidean Jordan algebra
  (or {\bf 2SNPE-Jordan algebra} for short), $\hk$ be another
  2SN-Jordan algebra and $\pi: \hk \rightarrow \End_{s}(\J)$ be a
  linear map. Consider the bilinear map $\varphi: \J\times \J
  \rightarrow \hk^*$ defined by $\varphi(x,y)(z) = B(\pi(z)(x),y),
  \forall x,y \in \J, z\in \hk$. Let \[\Jb = \hk\oplus \J \oplus
  \hk^*.\] Define the following product on $\Jb$:
  \[ (x+y+f)(x'+y'+f') = xx' + yy' +\pi(x)(y') + \pi(x')(y) + f'\circ
  R_{x} + f\circ R_{x'} + \varphi(y,y') \] for all $x, x' \in \hk, y,
  y'\in \J, f, f' \in \hk^*$. Then $\Jb$ is a 2SN-Jordan algebra if
  and only if $\pi$ is a 2SN-admissible representation of $\hk$ in
  $\J$. Moreover, $\Jb$ is pseudo-Euclidean with the bilinear form
\[
\overline{B}(x+y+f,x'+y'+f') = B(y,y') + f(x') + f'(x), \forall x, x'
\in \hk, y, y'\in \J, f, f' \in \hk^*.
\]

In this case, we say that $\Jb$ is a {\bf 2-step nilpotent double
  extension} (or {\bf 2SN-double extension}) of $\J$ by $\hk$ by means
of $\pi$.
\end{prop}

\begin{proof}
  If $\Jb$ is 2-step nilpotent then the product is commutative and
  $((x+y+f)(x'+y'+f'))(x''+y''+f'') = 0$ for all $x,x',x'' \in \hk, y,
  y', y'' \in \J, f, f',f'' \in \hk^*$. By a straightforward
  computation, one has that $\pi$ is a 2SN-admissible representation
  of $\hk$ in $\J$.

  Conversely, assume that $\pi$ is a 2SN-admissible representation of
  $\hk$ in $\J$. First, we set the extension $\J\oplus\hk^*$ of $\J$
  by $\hk^*$ with the product:
  \[(y+f)(y'+f') = yy'+\varphi(y,y'),\ \forall y,y'\in\J, \
  f,f'\in\hk^*.\] Since $\pi(z)\in\End_{s}(\J)$ and $\pi(z)(yy') = 0,\
  \forall z\in\hk, y,y'\in\J,$ then one has $\varphi$ symmetric and
  $\varphi(yy',y'') = 0$ for all $y,y',y''\in \J$. By Definition
  \ref{defn3.2}, $\J\oplus\hk^*$ is a 2SN-central extension of $\J$ by
  $\hk^*$.

  Next, we consider the direct sum $\pi\oplus R^*$ of two
  representations: $\pi$ and $R^*$ of $\hk$ in $\J\oplus\hk^*$ (see
  Definition \ref{defn3.6}). By a straightforward computation, we
  check that $\pi\oplus R^*$ satisfies the conditions of Proposition
  \ref{prop3.7} then the semi-direct sum of $\J \oplus \hk^*$ by $\hk$ by means
  of $\pi \oplus R^*$ is 2-step nilpotent. Finally, the product
  defined in $\Jb$ is exactly the product defined by the semi-direct
  sum in Proposition \ref{prop3.7}. Therefore we obtain the necessary
  and sufficient conditions.

  As a consequence of Definition \ref{defde}, $\overline{B}$ is an
  associative scalar product of $\Jb$, then $\Jb$ is a 2SNPE-Jordan
  algebra.
\end{proof}

The notion of 2SN-double extension {\bf does not characterize} all
2SNPE-Jordan algebras: there exist 2SNPE-Jordan algebras that can be
not described in term of 2SN-double extensions, for example, the
2SNPE-Jordan algebra $\J = \CC a \oplus \CC b$ with $a^2 =b$ and
$B(a,b) = 1$, zero otherwise. Therefore, we need a better
characterization given by the Proposition below, its proof is a matter
of a simple calculation.

\begin{prop}
  Let $(\J, B)$ be a 2SNPE-Jordan algebra, $(D,x_{0})\in \End_s(\J)
  \times \J$ be a 2SN-admissible pair with $B(x_0,x_0)=0$ and
  $(\tk=\CC x_1\oplus \CC y_1,B_\tk)$ be a quadratic vector space
  satisfying
\[B_\tk(x_1,x_1)=B_\tk(y_1,y_1)=0, \ B_\tk(x_1,y_1) = 1.
\]
Fix $\alpha$ in $\CC$ and consider the vector space \[\Jb =
\J\oplusp\tk\] equipped with the product
\[y_1\star y_1 = x_0 + \alpha x_1, \ y_1\star x = x\star y_1 = D(x)+
B(x_0,x)x_1, \ x\star y = xy + B(D(x),y)x_1
\]
and $x_1\star \Jb = \Jb\star x_1 = 0, \forall x,y\in \J$. Then $\Jb$
is a 2SNPE-Jordan algebra with the bilinear form $\overline{B} = B +
B_\tk$.

In this case, $(\Jb,\overline{B})$ is called a {\bf generalized double
  extension} of $\J$ by means of $(D,x_0,\alpha)$.
\end{prop}

\begin{prop}\label{prop3.11}
  Let $(\J, B)$ be a 2SNPE-Jordan algebra. If $\J$ is non-Abelian then
  it is obtained from an Abelian algebra by a sequence of generalized
  double extensions.
\end{prop}

\begin{proof}
  Assume that $(\J, B)$ is a 2SNPE-Jordan algebra and $\J$ is non-Abelian. By
  Proposition \ref{prop-red}, $\J$ has a reduced ideal $\lk$ that is
  still 2-step nilpotent. That means $\lk^2 \neq \lk$, so
  $\Ann(\lk)\neq \{0\}$. Therefore, we can choose nonzero $x_1\in
  \Ann(\lk)$ such that $B(x_1,x_1)=0$. Then there exists an isotropic
  element $y_1\in\J$ such that $B(x_1,y_1)=1$. Let $\J = (\CC
  x_1\oplus \CC y_1)\oplusp W$, where $W = (\CC x_1\oplus \CC
  y_1)^\bot$. We have that $\CC x_1$ and $x_1^\bot = \CC x_1\oplus W$
  are ideals of $\J$ as well.

  Let $x,y\in W$, $xy = \beta(x,y)+\alpha(x,y)x_1$, where
  $\beta(x,y)\in W$ and $\alpha(x,y)\in \CC$. It is easy to check that
  $W$ with the product $W\times W\rightarrow W$, $(x,y)\mapsto
  \beta(x,y)$ is a 2SN-Jordan algebra. Moreover, it is also
  pseudo-Euclidean with the bilinear form $B_W = B\ |_{W\times W}$.

  Now, we show that $\J$ is a generalized double extension of
  $(W,B_W)$. Indeed, let $x\in W$ then $y_1 x = D(x)+\varphi(x)
  x_1$, where $D$ is an endomorphism of $W$ and $\varphi \in
  W^*$. Since $y_1 (y_1 x) = y_1 (xy)= (y_1 x)y = 0, \forall x,y\in W$
  we get $D^2 (x) = D(x)y = D(xy) = 0, \forall x,y\in W$. Moreover,
  $B(y_1 x,y) = B(x,y_1 y) = B(y_1 ,xy),\forall x,y\in W$ implies that
  $D\in \End_s(W)$ and $\alpha(x,y) = B_W (D(x),y), \forall x,y\in W$.

  Since $B_W$ is non-degenerate and $\varphi \in W^*$ then there
  exists $x_0\in W$ such that $\varphi = B_W(x_0,.)$. Assume that
  $y_1y_1 = \mu y_1 + y_0 + \lambda x_1$. The equality $B(y_1y_1,x_1)
  = 0$ implies $\mu = 0$. Moreover, $y_0 = x_0$ since $B(y_1x,y_1) =
  B(x,y_1y_1),\forall x\in W$. Finally, $D(x_0) = 0$ is obtained by
  $y_1^3=0$ and this is enough to conclude that $\J$ is a generalized
  double extension of $(W,B_W)$ by means of $(D,x_0,\lambda)$.
\end{proof}

\subsection{$T^*$-extensions of pseudo-Euclidean 2-step nilpotent}
\hfill

Given a 2SN-Jordan algebra $\J$ and a symmetric bilinear map
$\theta:\J\times\J\rightarrow \J^*$ such that $R^*(z)(\theta(x,y))+\theta(xy,z) = 0,\ \forall x,y,z\in\J$, then by Proposition
\ref{prop-un}, $\J\oplus\J^*$ is also a 2SN-algebra. Moreover, if
$\theta$ is cyclic (that is, $\theta(x,y)(z) = \theta(y,z)(x), \forall
x,y,z\in\J$), then $\Jb$ is a pseudo-Euclidean Jordan algebra with the
bilinear form defined by
\[B(x+f,y+g) = f(y)+g(x), \ \forall x,y\in \J, f,g\in\J^*.
\]

In a more general framework, we can define:

\begin{defn}
  Let $\ak$ be a complex vector space and $\theta:
  \ak\times\ak\rightarrow \ak^*$ a cyclic symmetric bilinear map. Assume
  that $\theta$ is non-degenerate, i.e. if $\theta(x,\ak)=0$ then $x =
  0$. Consider the vector space $\J:= \ak\oplus \ak^*$ equipped the
  product
\[(x+f)(y+g) = \theta(x,y)\]
and the bilinear form
\[B(x+f,y+g) = f(y)+g(x)\] for all $x+f,y+g\in \J$. Then $(\J,B)$ is a
2SNPE-Jordan algebra and it is called the {\em $T^*$-extension} of
$\ak$ by $\theta$.
\end{defn}

\begin{lem}
  Let $\J$ be a $T^*$-extension of $\ak$ by $\theta$. If $\J\neq
  \{0\}$ then $\J$ is reduced.
\end{lem}
\begin{proof}
  Since $\theta$ is non-degenerate, it is easy to check that $\Ann(\J)
  = \ak^*$ is totally isotropic by the above definition.
\end{proof}
\begin{prop}\label{prop3.14}
  Let $(\J,B)$ be a 2SNPE-Jordan algebra. If $\J$ is reduced then $\J$
  is isometrically isomorphic to some $T^*$-extension.
\end{prop}
\begin{proof}
  Assume $\J$ is a reduced 2SNPE-Jordan algebra. Then one has
  $\Ann(\J) = \J^2$, so $\dim(\J^2) = \frac{1}{2}\dim(\J)$. Let $\J =
  \Ann(\J)\oplus \ak$, where $\ak$ is a complementary subspace of
  $\Ann(\J)$ in $\J$. Then $\ak\cong \J/\J^2$ as an Abelian
  algebra. Since $\ak$ and $\Ann(\J)$ are maximal totally isotropic
  subspaces of $\J$, we can identify $\Ann(\J)$ to $\ak^*$ by the
  isomorphism $\varphi:\Ann(\J)\rightarrow \ak^*$, $\varphi(x)(y) =
  B(x,y),\forall x\in \Ann(\J), y\in\ak$. Define $\theta:
  \ak\times\ak\rightarrow \ak^*$ by $\theta(x,y)=\varphi(xy),\forall
  x,y\in \ak$.

  Now, set $\alpha: \J\rightarrow \ak\oplus \ak^*$ by $\alpha(x) =
  p_1(x) + \varphi(p_2(x)),\forall x\in \J$, where $p_1: \J\rightarrow
  \ak$ and $p_2: \J\rightarrow \Ann(\J)$ are canonical
  projections. Then $\alpha$ is isometrically isomorphic.
\end{proof}
\begin{prop}\label{prop3.15}
  Let $\J_1$ and $\J_2$ be two $T^*$-extensions of $\ak$ by $\theta_1$
  and $\theta_2$ respectively. Then:

\begin{enumerate}

\item there exists a Jordan algebra isomorphism between $\J_1$ and
  $\J_2$ if and only if there exist an isomorphism $A_1$ of $\ak$ and
  an isomorphism $A_2$ of $\ak^*$ satisfying:
  \[A_2(\theta_1(x,y)) = \theta_2(A_1(x),A_1(y)), \forall x,y\in\ak.\]

\item there exists a Jordan algebra i-isomorphism between $\J_1$ and
  $\J_2$ if and only if there exists an isomorphism $A_1$ of $\ak$
\[\theta_1(x,y) = \theta_2(A_1(x),A_1(y))\circ A_1, \forall x,y\in\ak.\]
\end{enumerate}
\end{prop}

\begin{proof}\hfill
\begin{enumerate}
\item Let $A: \J_1\rightarrow \J_2$ be a Jordan algebra
  isomorphism. Since $\ak^* = \Ann(\J_1) = \Ann(\J_2)$ is stable by
  $A$ then there exist linear maps $A_1:\ak\rightarrow \ak$,
  $A'_1:\ak\rightarrow \ak^*$ and $A_2:\ak^*\rightarrow \ak^*$ such
  that:
  \[A(x+f) = A_1(x) + A'_1(x) + A_2(f), \ \forall x+f\in \J_1.
\]
Since $A$ is an isomorphism one has $A_2$ also isomorphic. We show that
$A_1$ is an isomorphism of $\ak$. Indeed, if $A_1(x_0) = 0$ with some
$x_0\in\ak$ then $A(x_0) = A'_1(x_0)$ and
\[ 0=A(x_0)\J_2 = A(x_0A^{-1}(\J_2)) = A(x_0\J_1).\] That implies
$x_0\J_1 = 0$ and so $x_0\in\ak^*$. That means $x_0=0$, i.e. $A_1$ is
an isomorphism of $\ak$.

For all $x$ and $y\in\ak$, one has $A(xy) = A(\theta_1(x,y)) =
A_2(\theta_1(x,y))$ and
\[A(x)A(y) = (A_1(x)+A'_1(x))(A_1(y)+A'_1(y)) = A_1(x)A_1(y) =
\theta_2(A_1(x),A_1(y)).\] Therefore, $A_2(\theta_1(x,y)) =
\theta_2(A_1(x),A_1(y)), \forall x,y\in\ak$.

Conversely, if there exist an isomorphism $A_1$ of $\ak$ and an
isomorphism $A_2$ of $\ak^*$ satisfying:
\[A_2(\theta_1(x,y)) = \theta_2(A_1(x),A_1(y)), \forall x,y\in\ak,\]
then we define $A:\J_1\rightarrow \J_2$ by $A(x+f) =
A_1(x)+A_2(f),\forall x+f\in\J_1$. It is easy to see that $A$ is a
Jordan algebra isomorphism.
\item Assume $A: \J_1\rightarrow \J_2$ is a Jordan algebra
  i-isomorphism then there exist $A_1$ and $A_2$ defined as in
  (1). Let $x\in\ak, f\in\ak^*$, one has:
  \[B'(A(x),A(f)) = B(x,f) \Rightarrow A_2(f)(A_1(x)) = f(x).
\]
Hence, $A_2(f) = f\circ A_1^{-1}, \forall f\in\ak^*$. Moreover,
$A_2(\theta_1(x,y)) = \theta_2(A_1(x),A_1(y))$ implies that
\[\theta_1(x,y)) = \theta_2(A_1(x),A_1(y))\circ A_1, \forall x,y\in\ak.\]

Conversely, define $A(x+f) = A_1(x)+ f\circ A_1^{-1},\forall
x+f\in\J_1$ then $A$ is an i-isomorphism.
\end{enumerate}

\end{proof}

\begin{ex}\label{ex3.18}
  We keep the notations as above. Let $\J'$ be the $T^*$-extension of
  $\ak$ by $\theta' = \lambda\theta,\lambda\neq 0$ then $\J$ and $\J'$
  is i-isomorphic by $A: \J\rightarrow \J'$ defined by
  \[ A(x+f) = \frac{1}{\alpha} x+\alpha f,\forall x+f\in\J.\] where
  $\alpha\in\CC,\ \alpha^3 = \lambda$.
\end{ex}

\bigskip

For a non-degenerate cyclic symmetric map $\theta$ of $\ak$, define a
trilinear form
\[I(x,y,z) = \theta(x,y)z, \forall x,y,z\in \ak.
\]
Then $I\in \Sb ^3(\ak)$, the space of symmetric trilinear forms on
$\ak$. The non-degenerate condition of $\theta$ is equivalent to
$\frac{\partial I}{\partial p}\neq 0, \forall p\in \ak^*$.

Conversely, let $\ak$ be a complex vector space and $I\in \Sb ^3(\ak)$
such that $\frac{\partial I}{\partial p}\neq 0$ for all $p\in \ak^*$. Define $\theta:
\ak \times \ak \rightarrow \ak^*$ by $\theta(x,y):= I(x,y,.), \forall
x,y\in \ak$ then $\theta$ is symmetric and non-degenerate. Moreover,
since $I$ is symmetric, then $\theta$ is cyclic and we obtain a
reduced 2SNPE-Jordan algebra
$T_{\theta}^*(\ak)$ defined by $\theta$. Therefore, there is a
one-to-one map from the set of all $T^*$-extensions of a complex
vector space $\ak$ onto the subset $\{I\in \Sb ^3(\ak) \mid
\frac{\partial I}{\partial p}\neq 0, \forall p\in \ak^*\}$, such
elements are also called {\em non-degenerate}.

\begin{cor}
  Let $\J_1$ and $\J_2$ be $T^*$-extensions of $\ak$ with respect to
  $I_1$ and $I_2$ non-degenerate. Then $\J$ and $\J'$ are i-isomorphic
  if and only if there exists an isomorphism $A$ of $\ak$ such that
\[I_{1}(x,y,z) = I_{2}(A(x), A(y), A(z)), \forall x,y,z\in \ak.\]
\end{cor}

In particular, $\J$ and $\J'$ are i-isomorphic if and only if there is
a isomorphism $^t{}A$ on $\ak^*$ which induces the isomorphism on $\Sb
^3(\ak)$, also denoted by $^t{}A$ such that $^t{}A(I_1)=I_2$. In this
case, we say that $I_1$ and $I_2$ are {\em equivalent}.

\begin{ex}
  Let $\ak = \CC a$ be one-dimensional vector space then $\Sb^3(\ak) =
  \CC (a^*)^3$. By Example \ref{ex3.18}, $T^*$-extensions of $\ak$ by
  $(a^*)^3$ and $\lambda (a^*)^3,\ \lambda\neq 0$, are i-isomorphic
  (also, these trilinear forms are equivalent). Hence, there is only
  one i-isomorphic class of $T^*$-extensions of $\ak$, that is $\J =
  \CC a \oplus \CC b$ with $a^2 =b$ and $B(a,b) = 1$, the other are
  zero.
\end{ex}

\bigskip

Now, let $\ak = \CC x\oplus\CC y$ be a 2-dimensional vector space then
\[\Sb^3(\ak) = \{a_1(x^*)^3 + a_2(x^*)^2y^* + a_3x^*(y^*)^2 +
a_4(y^*)^3, a_i\in\CC.\]

It is easy to prove that every bivariate homogeneous polynomial of
degree 3 is reducible. Therefore, by a suitable basis choice
(certainly isomorphic), a non-degenerate element $I\in\Sb^3(\ak)$ has the
form $I = ax^*y^*(bx^*+ cy^*),\ a,b\neq 0$. Replace $x^*:=\alpha x^*$
with $\alpha^2 = ab$ to get the form $I_{\lambda} = x^*y^*(x^*+
\lambda y^*),\ \lambda\in\CC$.

Next, we will show that $I_0$ and $I_\lambda,\lambda \neq 0$ are not
equivalent. Indeed, assume the contrary, i.e. there is an isomorphism
$^t{}A$ such that $^t{}A(I_0) = I_\lambda$. We can write
\[^t{}A(x^*) = a_1x^*+ b_1y^*,\ ^t{}A(y^*) = a_2x^*+ b_2y^*, \ a_1,a_2,b_1,b_2\in\CC.\]

Then
\begin{eqnarray*}
  ^t{}A(I_0)=(a_1x^*+ b_1y^*)^2(a_2x^*+ b_2y^*)
  = a_1^2a_2(x^*)^3 + (a_1^2b_2 + 2a_1a_2b_1)(x^*)^2y^* + \\ (2a_1b_1b_2 + a_2b_1^2)x^*(y^*)^2 + b_1^2b_2(y^*)^3.
\end{eqnarray*}

Comparing the coefficients we will get a contradiction. Therefore,
$I_0$ and $I_\lambda,\lambda \neq 0$ are not equivalent.

However, two forms $I_{\lambda_1}$ and $I_{\lambda_2}$ where $\lambda_1,\lambda_2\neq 0$ are equivalent
by the isomorphism $^t{}A$ satisfying $^t{}A(I_{\lambda_1}) = I_{\lambda_2}$ defined by:
\[ ^t{}A(x^*) =\alpha y^*,\ ^t{}A(y^*) = \beta x^*
\] where $\alpha,\beta\in\CC$ such that $\alpha^3 =
\lambda_1\lambda_2^2$, $\beta^3 = \frac{1}{\lambda_1^2\lambda_2}$. This implies
that there are only two i-isomorphic classes of $T^*$-extensions of
$\ak$.

\begin{ex}
  Let $\J_0 = \spa\{x,y,e,f\}$ be a $T^*$-extension of a 2-dimensional
  vector space $\ak$ by $I_0 = (x^*)^2y^*$, with $e = x^*$ and $f=
  y^*$, that means $B(x,e) =B(y,f) = 1$, the other are zero. It is
  easy to compute the product in $\J_0$ defined by $x^2 = f$,
  $xy=e$. Let $I_\lambda= x^*y^*(x^*+ \lambda y^*), \lambda \neq 0$
  and $\J_\lambda = \spa\{x,y,e,f\}$ be another $T^*$-extension of the
  2-dimensional vector space $\ak$ by $I_\lambda$. The products on
  $\J_\lambda$ are $x^2 = f$, $xy = e + \lambda f$ and $yy = \lambda
  e$. These two algebras are neither i-isomorphic nor
  isomorphic. Indeed, if there is $A: \J_0\rightarrow \J_\lambda$ an
  isomorphism. Assume $A(y) = \alpha_1x+\alpha_2y +
  \alpha_3e+\alpha_4f$ then
  \[0 = A(yy)= ( \alpha_1x+\alpha_2y + \alpha_3e+\alpha_4f)^2 =
  \alpha_1^2x^2 + 2\alpha_1\alpha_2xy + \alpha_2^2y^2.\] We obtain
  $(\lambda\alpha_2^2+2\alpha_1\alpha_2)e +
  (2\lambda\alpha_1\alpha_2+\alpha_1^2)f = 0$. Hence, $\alpha_1 =
  \pm\lambda\alpha_2$. Both cases imply $\alpha_1 =\alpha_2 =0$
  (a contradiction).

  We can also conclude that there are only two isomorphic classes of
  $T^*$-extensions of $\ak$.
\end{ex}

\section{Symmetric Novikov algebras}

\begin{defn}\label{defn4.1}
  An algebra $\N$ over $\CC$ with a bilinear product
  $\N\times\N\rightarrow \N$, $(x, y)\mapsto xy$ is called a {\em
    left-symmetric algebra} if it satisfies the identity:
\begin{equation} \label{nov-1}
(xy)z - x(yz) = (yx)z - y(xz), \forall x, y, z \in \N.
\end{equation}
or in terms of associators
\begin{equation*} (x,y,z) = (y,x,z), \forall x, y, z \in \N.
\end{equation*}
It is called a {\em Novikov algebra} if in addition
\begin{equation} \label{nov-2} (xy)z = (xz)y
\end{equation}
holds for all $x, y, z \in \N$. In this case, the commutator $[x, y]
:= xy - yx$ of $\N$ defines a Lie algebra, denoted by $\g (\N)$, which
is called the {\em sub-adjacent Lie algebra} of $\N$. It is known that
$\g (\N)$ is a solvable Lie algebra \cite{Bur06}. Conversely, let $\g$
be a Lie algebra with Lie bracket $[.,.]$. If there exists a bilinear
product $\g \times \g \rightarrow \g, (x,y)\mapsto xy$ that satisfies
(\ref{nov-1}), (\ref{nov-2}) and $[x, y] = xy - yx,\forall x,y\in\J$
then we say that $\g$ admits a {\em Novikov structure}.
\end{defn}

\begin{ex}\label{ex4.2}
  Every 2-step nilpotent algebra $\N$ satisfying $(xy)z=x(yz) = 0,
  \forall x,y,z\in\N$, is a Novikov algebra.
\end{ex}

For $x \in \N$, denote by $L_x$ and $R_x$ respectively the left and
right multiplication operators $L_x (y) = xy$, $R_x(y)= yx$, $\forall
y\in \N$. The condition (\ref{nov-1}) is equivalent to $[L_x, L_y] =
L_{[x, y]}$ and (\ref{nov-2}) is equivalent to $[R_x, R_y] =
0$. In the other words, the left-operators form a Lie algebra and the right-operators commute.

It is easy to check two Jacobi-type identities:
\begin{prop}
  Let $\N$ be a Novikov algebra then for all $x,y,z\in \N$:
\[[x,y]z + [y,z]x + [z,x]y = 0,
\]
\[x[y,z] + y[z,x] + z[x,y] = 0.
\]
\end{prop}
\begin{defn}
  Let $\N$ be a Novikov algebra. A bilinear form $B: \N \times \N
  \rightarrow \CC$ is called {\em associative} if \[B(xy, z) = B(x,
  yz), \forall x,y,z \in \N.\] We say that $\N$ is a {\em symmetric
    Novikov algebra} if it is endowed a non-degenerate associative
  symmetric bilinear form $B$.
\end{defn}



Let $(\N , B)$ be a symmetric Novikov algebra and $S$ be a subspace of
$\N$. Denote by $S^\bot$ the set $\{x \in \N \mid B(x,S) = 0\}$. If
$B|_{S\times S}$ is non-degenerate (resp. degenerate) then we say that
$S$ is {\em non-degenerate} (resp. {\em degenerate}).

The proof of Lemma \ref{lem4.6} and Proposition \ref{4.4.5} below is lengthy, but straight forward then we omit it.
\begin{lem}\label{lem4.6}
Let $(\N , B)$ be a symmetric Novikov algebra and $I$ be an ideal of $\N$ then

\begin{enumerate}
\item $I^\bot$ is also an ideal of $\N$ and $II^\bot = I^\bot I =
  \{0\}$

\item If $I$ is non-degenerate then so is $I^\bot$ and $\N = I \oplusp
  I^\bot$.
\end{enumerate}
\end{lem}



\begin{prop} \label{4.4.5} We call the set $C(\N):= \{x\in \N \mid xy
  =yx, \forall y \in \N \}$ the {\bf center} of $\N$ and denote by
  $As(\N) = \{x\in\N\ \mid (x,y,z)=0,\forall y,z\in\N\}$. One has

\begin{enumerate}
\item If $\N$ is a Novikov algebra then $C(\N)\subset
  N(\N)$, where $N(\N)$ is the nucleus of $\N$ defined in Definition \ref{1.11} (3). Moreover, if $\N$ is also commutative then
  $N(\N) = \N = As(\N)$ (that means $\N$ is an associative algebra).
 \item If $(\N, B)$ is a symmetric Novikov algebra then
  \begin{itemize}
    \item[(i)] $C(\N)=[\g(\N), \g(\N)]^\bot$.
    \item[(ii)] $N(\N) = As(\N) = (\N, \N, \N)^\bot$.
    \item[(iii)] $L\Ann(\N) = R\Ann(\N) = \Ann(\N) = (\N \N)^\bot$.
  \end{itemize}
  \end{enumerate}
\end{prop}






\begin{prop}
Let $\N$ be a Novikov algebra then
\begin{enumerate}
\item $C(\N)$ is a commutative subalgebra.

\item $As(\N)$, $N(\N)$ are ideals.
\end{enumerate}
\end{prop}

\begin{proof}\hfill
\begin{enumerate}
\item Let $x,y\in C(\N)$ then $(xy)z = (xz)y=(zx)y = z(xy)
  +(z,x,y)=z(xy),\forall z\in\N$. Therefore, $xy\in C(\N)$ and then
  $C(\N)$ is a subalgebra of $\N$. Certainly, $C(\N)$ is commutative.
\item Let $x\in As(\N),y,z,t\in\N$. By the equality
  \[(xy,z,t) = ((xy)z)t-(xy)(zt) = ((xz)t)y-(x(zt))y = (x,z,t)y = 0,\]
  one has $xy\in As(\N)$. Moreover,
  \[(yx,z,t) = ((yx)z)t-(yx)(zt) = (y(xz))t-y(x(zt)) \]
  \[= (y,xz,t)+y((xz)t) -y(x(zt)) = y(x,z,t) = 0\] since $xz\in
  As(\N)$. Therefore $As(\N)$ is an ideal of $\N$.

Similarly, let $x\in N(\N),y,z,t\in\N$ one has:
\[(y,z,xt) = (yz)(xt)-y(z(xt)) = ((yz)x)t-(yz,x,t)-y((zx)t -(z,x,t))
\]
\[=((yz)x)t - (y(zx))t + (y,zx,t) = (y,z,x)t = 0\] and \[(y,z,tx) =
(yz)(tx)-y(z(tx)) = ((yz)t)x-(yz,t,x)-y((zt)x-(z,t,x))\]
\[=((yz)x)t-y((zx)t) = (y,z,x)t +(y,zx,t) = 0.\] Then $N(\N)$ is also
an ideal of $\N$.
\end{enumerate}

\end{proof}

\begin{lem}\label{4.4.6}
  Let $(\N,B)$ be a symmetric Novikov algebra then $[L_x,L_y] =
  L_{[x,y]} = 0$ for all $x,y\in \N$. Consequently, for a symmetric Novikov algebra, the Lie algebra formed by the left-operators is Abelian.
\end{lem}

\begin{proof}
  It follows the proof of Lemma II.5 in \cite{AB10}. Fix $x,y\in\N$,
  for all $z,t\in\N$ one has
  \[B([L_x,L_y](z),t) = B(x(yz)-y(xz),t) = B((tx)y - (ty)x,z) = 0.\]
  Therefore, $[L_x,L_y] = L_{[x,y]} = 0,\forall x,y\in\N$.
\end{proof}

\begin{cor}\label{cor4.10}
  Let $(\N,B)$ be a symmetric Novikov algebra then the sub-adjacent Lie algebra $\g (\N)$ of $\N$ with the bilinear form
  $B$ becomes a quadratic 2-step nilpotent Lie algebra.
\end{cor}

\begin{proof}
  One has
  \[ B([x,y],z) = B(xy - yx, z)= B(x,yz) - B(x,zy) =
  B(x,[y,z]),\forall x,y,z\in \N.
\]
Hence, $\g (\N)$ is quadratic. By Lemma \ref{4.4.6} and 2(iii) of
Proposition \ref{4.4.5} one has $[x,y]\in L\Ann(\N)=\Ann(\N)$, $\forall x,y\in
\N$. That implies $[[x,y],z]\in \Ann(\N)\N = \{0\}, \forall x,y\in
\N$, i.e. $\g (\N)$ is 2-step nilpotent.
\end{proof}

It results that the classification of quadratic 2-step nilpotent Lie
algebras (\cite{Ova07}, \cite{Duo10}) is closely related to the
classification of symmetric Novikov algebras. For instance, by
\cite{DPU10}, every quadratic 2-step nilpotent Lie algebra of
dimension $\leq5$ is Abelian so that every symmetric Novikov algebra
of dimension $\leq5$ is commutative. In general, in the case of
dimension $\geq 6$, there exists a non-commutative symmetric Novikov
algebra by Proposition \ref{4.4.9} below.

\begin{defn}
  Let $\N$ be a Novikov algebra. We say that $\N $ is an {\em
    anti-commutative Novikov algebra} if \[xy = -yx, \forall x,y \in
  \N.\]
\end{defn}
\begin{prop}\label{4.4.9}
  Let $\N$ be a Novikov algebra. Then $\N$ is anti-commutative if, and
  only if, $\N$ is a 2-step nilpotent Lie algebra with the Lie bracket
  defined by $[x, y]:= xy, \forall x, y\in \N$.
\end{prop}
\begin{proof}
  Assume that $\N$ is a Novikov algebra such that $xy = -yx, \forall
  x,y \in \N$. Since the commutator $[x,y] = xy -yx = 2xy$ is a Lie
  bracket, so the product $(x,y) \mapsto xy$ is also a Lie
  bracket. The identity (\ref{nov-1}) of Definition \ref{defn4.1} is equivalent to $(xy)z = 0, \forall x,y,z
  \in \N$. It shows that $\N$ is a 2-step nilpotent Lie algebra.

  Conversely, if $\N$ is a 2-step nilpotent Lie algebra then we define
  the product $xy:= [x,y], \forall x,y \in \N$. It is obvious that the
  identities (\ref{nov-1}) and (\ref{nov-2}) of Definition \ref{defn4.1} are satisfied since $(xy)z = 0, \forall x,y,z \in \N$.
\end{proof}

By the above Proposition, the study of anti-commutative Novikov
algebras is reduced to the study of 2-step nilpotent Lie
algebras. Moreover, the formula in this proposition also can be used
to define a 2-step nilpotent symmetric Novikov algebra from a
quadratic 2-step nilpotent Lie algebra. Recall that there exists only
one non-Abelian quadratic 2-step nilpotent Lie algebra of dimension 6
up to isomorphism \cite{DPU10} then there is only one anti-commutative
symmetric Novikov algebra of dimension 6 up to isomorphism. However,
there exist non-commutative symmetric Novikov algebras that are not
2-step nilpotent \cite{AB10}. For instance, let $\N = \g_6\oplusp \CC
c$, where $\g_6$ is the 6-dimensional elementary quadratic Lie algebra
\cite{DPU10} and $\CC c$ is a pseudo-Euclidean simple Jordan algebra
with the bilinear form $B_c(c,c)=1$ (obviously, this algebra is a
symmetric Novikov algebra and commutative). Then $\N $ become a
symmetric Novikov algebra with the bilinear form defined by $B =
B_{\g_6} + B_c$, where $B_{\g_6}$ is the bilinear form on $\g_6$. We
can extend this example for the case $\N = \g\oplusp \J$, where $\g$
is a quadratic 2-step nilpotent Lie algebra and $\J$ is a symmetric
Jordan-Novikov algebra defined below. However, these algebras are
decomposable. An example in the indecomposable case of dimension 7 can
be found in the last part of this paper.

\begin{prop}\label{prop4.13}
  Let $\N$ be a Novikov algebra. Assume that its product is
  commutative, that means $xy = yx, \forall x,y \in \N$. Then the
  identities (\ref{nov-1}) and (\ref{nov-2}) of Definition \ref{defn4.1} are equivalent to the only condition:
  \[(x,y,z) = (xy)z - x(yz) = 0, \forall x,y,z \in \N.
\]

It means that $\N$ is an associative algebra. Moreover, $\N$ is also
a Jordan algebra. In this case, we say that $\N$ is a {\bf
  Jordan-Novikov algebra}. In addition, if $\N$ has a non-degenerate
associative symmetric bilinear form, then we say that $\N$ is a {\em
  symmetric} Jordan-Novikov algebra.
\end{prop}

\begin{proof}
  Assume $\N$ is a commutative Novikov algebra. By (1) of Proposition
  \ref{4.4.5}, the product is also associative. Conversely, if
  one has the condition:
\[(xy)z - x(yz) = 0, \forall x,y,z \in \N
\]
then (\ref{nov-1}) identifies with zero and (\ref{nov-2}) is obtained by $(yx)z=y(xz),\forall x,y,z \in \N$.
\end{proof}

\begin{ex}\label{ex4.14}
  Recall the pseudo-Euclidean Jordan algebra $\J$ in Example \ref{2.9}
  spanned by $\{x,x_1,y_1\}$, where the commutative product on $\J$ is
  defined by:
\[ y_1^2=y_1, y_1x=x, y_1x_1 = x_1, x^2= x_1.
\]
It is easy to check that this product is also associative. Therefore,
$\J$ is a symmetric Jordan-Novikov algebra with the bilinear form $B$
defined $B(x_1,y_1) = B(x,x)=1$ and the other zero.
\end{ex}
\begin{ex} Pseudo-Euclidean 2-step nilpotent Jordan algebras are
  symmetric Jordan-Novikov algebras.
\end{ex}
\begin{rem}\hfill\label{rem4.16}
\begin{enumerate}
\item By Lemma \ref{4.4.6}, if the symmetric Novikov algebra $\N$ has
  $\Ann(\N) = \{0\}$ then $[x,y] = xy -yx = 0,\forall x,y\in\N$. It
  implies that $\N$ is commutative and then $\N$ is a symmetric
  Jordan-Novikov algebra.
    \item If the product on $\N$ is associative then it may not be
       commutative. An example can be found in the next part.
     \item Let $\N$ be a Novikov algebra with unit element $e$; that
       is $ex=xe=x,\forall x\in \N$. Then $xy = (ex)y = (ey)x = yx,\
       \forall x,y\in\N$ and therefore $\N$ is a Jordan-Novikov
       algebra.
     \item The algebra given in Example \ref{ex4.14} is also a
       Frobenius algebra, that is, a finite-dimensional associative
       algebra with unit element equipped with a non-degenerate
       associative bilinear form.
\end{enumerate}
\end{rem}

A well-known result is that every associative algebra $\N$ is
Lie-admissible and Jordan-admissible; that is, if $(x,y)\mapsto xy$ is
the product of $\N$ then the products \[ [x,y] = xy -yx \qquad \text{
  and } \] \[[x, y]_+ := xy+yx\] define respectively a Lie algebra
structure and a Jordan algebra structure on $\N$. There exist algebras
satisfying each one of these properties. For example, the
non-commutative Jordan algebras are Jordan-admissible \cite{Sch55} or
the Novikov algebras are Lie-admissible. However, remark that a
Novikov algebra may not be Jordan-admissible by the following example:

\begin{ex} \label{explus} Consider the 2-dimensional algebra $\N = \CC
  a\oplus \CC b$ such that $ba=-a$, zero otherwise. Then $\N$ is a
  Novikov algebra \cite{BMH02}. One has $[a,b] = a$ and $[a, b]_+ =
  -a$. For $x\in\N$, denote by $\ad_x^+$ the endomorphism of $\N$
  defined by $\ad_x^+(y)=[x,y]_+ = [y,x]_+,\ \forall y\in\N$. It is
  easy to see that
\[\ad_a^+ = \begin{pmatrix} 0 & -1 \\ 0 & 0 \end{pmatrix}\text{ and } \ \ad_b^+=\begin{pmatrix} -1 & 0 \\ 0 & 0 \end{pmatrix}.\]
Let $x=\lambda a + \mu b\in\N,\ \lambda,\mu\in\CC$, one has $[x,x]_+ =
-2\lambda\mu a$ and therefore:
\[\ad_x^+ = \begin{pmatrix} -\mu & -\lambda \\ 0 & 0 \end{pmatrix}\text{ and } \ \ad^+_{[x,x]_+}=\begin{pmatrix} 0 & 2\lambda\mu \\ 0 & 0 \end{pmatrix}.\]
Since $[\ad_x^+,\ad^+_{[x,x]_+}]\neq 0$ if $\lambda,\mu\neq 0$, then
$\N$ is not Jordan-admissible.
\end{ex}
We will give a condition for a Novikov algebra to be
Jordan-admissible as follows:

\begin{prop}\label{prop4.17}
Let $\N$ be a Novikov algebra satisfying
\begin{equation} \label{V}
(x,x,x) = 0,\forall x\in\N.
\end{equation}
Define on $\N$ the product $[x,y]_+ = xy+yx, \forall x,y\in \N$ then
$\N$ is a Jordan algebra with this product. In this case, it is called
the {\bf associated Jordan algebra} of $\N$ and denoted by $\J(\N)$.
\end{prop}

\begin{proof} Let
  $x,y\in\N$ then we can write $x^3 = x^2x = xx^2$. One has
\[[[x, y]_+,[x, x]_+]_+ = [xy+yx,2x^2]_+
\]
\[= 2(xy)x^2 + 2(yx)x^2 + 2x^2(xy) + 2x^2(yx)\]
\[= 2x^3y + 2(yx)x^2 + 2x^2(xy) + 2x^2(yx)\]
and
\[[x, [y, [x, x]_+]_+]_+ = [x,2yx^2+2x^2y]_+\]
\[=2x(yx^2)+2x(x^2y)+2(yx^2)x+2(x^2y)x\]
\[=2x(yx^2)+2x(x^2y) + 2(yx)x^2 + 2x^3y.\] Therefore, $[[x, y]_+,[x,
x]_+]_+ = [x, [y, [x, x]_+]_+]_+$ if and only if $x^2(xy) + x^2(yx) =
x(yx^2)+x(x^2y)$. Remark that we have following identities:
\[x^2(xy) = x^3y -(x^2,x,y) = x^3y -(x,x^2,y),\]
\[x^2(yx) = (x^2y)x - (x^2,y,x) =  x^3y -(y,x^2,x),\]
\[x(yx^2) = (xy)x^2 -(x,y,x^2) = x^3y-(y,x,x^2),\]
\[x(x^2y) = x^3y - (x,x^2,y).\] It means that we have only to check
the formula $(y,x^2,x) = (y,x,x^2)$. It is clear by the identities
(\ref{nov-1}) and (\ref{V}). Then we can conclude that $\J(\N)$ is a
Jordan algebra.
\end{proof}

\begin{cor}
  If $(\N,B)$ is a symmetric Novikov algebra satisfying (\ref{V}) then
  $(\J(\N),B)$ is a pseudo-Euclidean Jordan algebra.
\end{cor}

\begin{proof}
  It is obvious since $B([x,y]_+,z) = B(xy+yx,z) = B(x,yz+zy) =
  B(x,[y,z]_+)$, $\forall x,y,z\in \J(\N)$.
\end{proof}

\begin{rem}\label{remplus} Obviously, Jordan-Novikov algebras are
  power-associative but in general this is not true for Novikov
  algebras. Indeed, if Novikov algebras were power-associative then
  they would satisfy (\ref{V}). That would imply they were
  Jordan-admissible. But, that is a contradiction as shown in Example
  \ref{explus}.
\end{rem}

\begin{lem}\label{4.4.16}
Let $\N$ be a Novikov algebra then $[x,yz]_+ = [y,xz]_+,\ \forall x,y,z\in\N$.
\end{lem}
\begin{proof}
  By (III), for all $x,y,z\in\N$ one has $(xy)z + y(xz) =
  x(yz)+(yx)z$. Combine with (IV), we obtain:
\[(xz)y + y(xz) = x(yz)+(yz)x.\] That means $[x,yz]_+ = [y,xz]_+,\
\forall x,y,z\in\N$.
\end{proof}

\begin{prop}\label{4.4.17}
Let $(\N,B)$ be a symmetric Novikov algebra then following identities:
\begin{enumerate}
  \item $x[y,z] = [y,z]x = 0$. Consequently, $[x,yz]_+ = [x,zy]_+$.
  \item $[x,y]_+z = [x,z]_+y$,
    \item  $[x,yz]_+= [xy, z]_+= x[y, z]_+ = [x,y]_+z$,
    \item $x[y,z]_+ = [y,z]_+x.$
\end{enumerate}
hold for all $x,y,z\in \N$.
\end{prop}

\begin{proof} Let $x,y,z,t$ be elements $\in \N$,
\begin{enumerate}
\item By Proposition \ref{4.4.5} and Lemma \ref{4.4.6}, $L_{[y,z]} = 0$
  so one has (1).
\item $B([x,y]_+z,t) = B(y,[x,zt]_+) = B(y,[z,xt]_+) =
  B([z, y]_+x,t)$. Therefore, $[x,y]_+z = [z, y]_+x$. Since the
  product $[.,.]_+$ is commutative then $[y,x]_+z = [y,z]_+x$.
    \item By (1) and Lemma \ref{4.4.16}, $[x,yz]_+ = [x,zy]_+ =
       [z,xy]_+ = [xy,z]_+$.

       Since $B$ is associative with respect to the product in $\N$
       and in $\J(\N)$ then
    \[B(t,[xy,z]_+)=B([t,xy]_+,z) = B([t,yx]_+,z) =
     B([y,tx]_+,z) = B(tx,[y,z]_+) = B(t,x[y,z]_+).\] It implies
     that $[xy,z]_+ = x[y,z]_+$. Similarly,
    \[B([x,y]_+z,t) = B(x,[y,zt]_+) = B(x,[y,tz]_+) =
     B(x,[t,yz]_+) = B([x,yz]_+,t).\] So $[x,y]_+z =
     [x,yz]_+$.
    \item By (2) and (3), $x[y, z]_+ = [x,y]_+z = [y,x]_+z = [y,z]_+x$.
\end{enumerate}

\end{proof}

\begin{cor} Let $(\N,B)$ be a symmetric Novikov algebra then
  $(\J(\N),B)$ is a symmetric Jordan-Novikov algebra.
\end{cor}
\begin{proof}
  We will show that $[[x,y]_+, z]_+ = [x,[y,z]_+]_+,\ \forall
  x,y,z\in \N$. Indeed, By Proposition \ref{4.4.17} one has
\[[[x,y]_+, z]_+ = [2xy,z]_+ = 2[z,xy]_+ = 2[x,yz]_+ =
[x,[y,z]_+]_+.\] Hence, the product $[.,.]_+$ are both commutative
and associative. That means $\J(\N)$ be a Jordan-Novikov algebra.
\end{proof}

It results that for symmetric Novikov algebras the condition (\ref{V}) is
not necessary. Moreover, we have the much stronger fact as follows:
\begin{prop}\label{prop4.22}
  Let $\N$ be a symmetric Novikov algebra then the product on $\N$ is
  associative, that is $x(yz) = (xy)z,\forall x,y,z\in\N$.
\end{prop}

\begin{proof}
Firstly, we need the lemma:
\begin{lem}\label{lem4.23}
Let $\N$ be a symmetric Novikov algebra then $\N\N\subset C(\N)$.
\end{lem}
\begin{proof}
  By Lemma \ref{4.4.6}, one has $[x,y] = xy-yx\in \Ann(\N)\subset
  C(\N),\forall x,y\in\N$. Also, by (4) of Proposition \ref{4.4.17},
  $x[y,z]_+ = [y,z]_+ x,\ \forall x,y,z\in\N$, that means $[x,y]_+ =
  xy+yx\in C(\N),\forall x,y\in C(\N)$. Hence, $xy\in C(\N),\forall
  x,y\in C(\N)$, i.e. $\N\N\subset C(\N)$.
\end{proof}
Let $x,y,z\in \N$. By above Lemma, one has $(yz)x = x(yz)$. Combine
with (\ref{nov-2}), $(yx)z=x(yz)$. On the other hand, $[x,y] \in \Ann(\N)$
implies $(yx)z=(xy)z$. Therefore, $(xy)z = x(yz)$.
\end{proof}

A general proof of the above Proposition can be found in \cite{AB10},
Lemma II.4 which holds for all symmetric left-symmetric superalgebras.

By Corollary \ref{cor4.10}, if $\N$ is a symmetric Novikov algebra
then $\g(\N)$ is 2-step nilpotent. However, $\J(\N)$ is not
necessarily 2-step nilpotent, for example the one-dimensional Novikov
algebra $\CC c$ with $c^2 = c$ and $B(c,c) = 1$. If $\N$ is a
symmetric 2-step nilpotent Novikov algebra then $(xy)z = 0, \forall
x,y,z\in\N$. So $[[x,y]_+,z]_+ = 0,\ \forall x,y,z\in\N$. That implies
$\J(\N)$ is also a 2-step nilpotent Jordan algebra. The converse is
also true.

\begin{prop}
  Let $\N$ be a symmetric Novikov algebra. If $\J(\N)$ is a 2-step
  nilpotent Jordan algebra then $\N$ is a 2-step nilpotent Novikov
  algebra.
\end{prop}
\begin{proof}
  Since (4) of Proposition \ref{4.4.17}, if $x,y,z\in\N$ then one has
  \[[[x,y]_+,z]_+= [x,y]_+z+z[x,y]_+ = 2[x,y]_+z = 0.\] It
  means $[x,y]_+ = xy+yx\in \Ann(\N)$. On the other hand, $[x,y] =
  xy-yx\in \Ann(\N)$ then $xy\in \Ann(\N),\forall
  x,y\in\N$. Therefore, $\N$ is 2-step nilpotent.
\end{proof}

By Proposition \ref{4.4.9}, since the lowest dimension of non-Abelian
quadratic 2-step nilpotent Lie algebras is six then examples of
symmetric non-commutative Novikov algebras must be at least six
dimensional. One of those can be found in \cite{ZC07} and it is also
described in term of double extension in \cite{AB10}. We recall this
algebra as follows:

\begin{ex} \label{ex4.25} Firstly, we define the {\bf character
    matrix} of a Novikov algebra $\N=\spa\{e_1,\dots,e_n\}$ by
\[\begin{pmatrix} \sum_k c_{11}^k e_k & \dots & \sum_k c_{1n}^k e_k \\ \vdots & \ddots & \vdots \\
    \sum_k c_{n1}^k e_k & \dots & \sum_k c_{nn}^k e_k \end{pmatrix},
\]

where $c_{ij}^k$ are the {\bf structure constants} of $\N$,
i. e. $e_ie_j= \sum_k c_{ij}^ke_k$.

  Now, let $\N_6$ be a 6-dimensional vector space spanned by
  $\{e_1,...,e_6\}$ then $\N_6$ is a symmetric non-commutative Novikov
  algebras with character matrix
  \[\begin{pmatrix} 0 & 0 & 0 & 0& 0 & 0 \\ 0 & 0 & 0 & 0& 0 & 0 \\ 0
    & 0 & 0 & 0& 0 & 0 \\ 0 & 0 & 0 & 0& e_3 & 0 \\ 0 & 0 & 0 & 0& 0 &
    e_1 \\ 0 & 0 & 0 & e_2 & 0 & 0 \end{pmatrix}
\]
and the bilinear form $B$ defined by:

\[\begin{pmatrix} 0 & 0 & 0 & 1& 0 & 0 \\ 0 & 0 & 0 & 0& 1 & 0 \\ 0 &
  0 & 0 & 0& 0 & 1 \\ 1 & 0 & 0 & 0& 0 & 0 \\ 0 & 1 & 0 & 0& 0 & 0 \\
  0 & 0 & 1 & 0 & 0 & 0 \end{pmatrix}.
\]

Obviously, in this case, $\N_6$ is a 2-step nilpotent Novikov algebra
with $\Ann(\N) = \N\N$. Moreover, $\N_6$ is indecomposable since it is
non-commutative and all of symmetric Novikov algebras up to dimension
5 are commutative.
\end{ex}

We need the following lemma:

\begin{lem}
  Let $\N$ be a non-Abelian symmetric Novikov algebra then $\N =
  \zk\oplusp\lk$, where $\zk\subset \Ann(\N)$ and $\lk$ is a reduced
  symmetric Novikov algebra, that means $\lk\neq \{0\}$ and
  $\Ann(\lk)\subset \lk\lk$.
\end{lem}

\begin{proof}
  Let $\zk_0 = \Ann(\N)\cap\N\N$, $\zk$ is a complementary subspace
  of $\zk_0$ in $\Ann(\N)$ and $\lk = (\zk)^\bot$. If $x$ is an
  element in $\zk$ such that $B(x,\zk) = 0$ then $B(x,\N\N) = 0$ since
  $\Ann(\N)=(\N\N)^\bot$. As a consequence, $B(x,\zk_0) = 0$ and then
  $B(x,\Ann(\N)) = 0$. Hence, $x$ must be in $\N\N$ since
  $\Ann(\N)=(\N\N)^\bot$. It shows that $x=0$ and $\zk$ is
  non-degenerate. By Lemma \ref{lem4.6}, $\lk$ is a non-degenerate ideal and
  $\N = \zk\oplusp\lk$.

  Since $\N$ is non-Abelian then $\lk\neq \{0\}$. Moreover, $\lk\lk =
  \N\N$ implies $\zk_0\subset \lk\lk$. It is easy to see that
  $\zk_0=\Ann(\lk)$ and the lemma is proved.
\end{proof}

\begin{prop}
  Let $\N$ be a symmetric non-commutative Novikov algebras of
  dimension 6 then $\N$ is 2-step nilpotent.
\end{prop}

\begin{proof}

  Let $\N = \spa\{x_1, x_2,x_3,z_1,z_2,z_3\}$. By \cite{DPU10}, there
  exists only one non-Abelian quadratic 2-step nilpotent Lie algebra
  of dimension 6 (up to isomorphisms) then $\g(\N) = \g_6$. We can
  choose the basis such that $[x_1,x_2] = z_3$, $[x_2,x_3] = z_1$,
  $[x_3,x_1] = z_2$ and the bilinear form $B(x_i,z_i) = 1,\ i=1,2,3$,
  the other are zero.

  Recall that $C(\N) := \{x\in\N\ \mid xy= yx,\forall y\in\N\}$ then
  $C(\N) = \{x\in\N\ |[x,y]=0,\forall y\in\N\}$. Therefore, $C(\N) =
  \spa\{z_1,z_2,z_3\}$ and $\N\N \subset C(\N)$ by Lemma
  \ref{lem4.23}. Consequently, $\dim(\N\N)\leq 3$.

  By the above lemma, if $\N$ is not reduced then $\N = \zk\oplusp\lk$
  with $\zk\subset \Ann(\N)$ is a non-degenerate ideal and $\zk\neq
  \{0\}$. It implies that $\lk$ is a symmetric Novikov algebra having
  dimension $\leq5$ and then $\lk$ is commutative. This is a
  contradiction since $\N$ is non-commutative. Therefore, $\N$ must be
  reduced and $\Ann(\N) \subset \N\N$. Moreover, $\dim(\N\N) +
  \dim(\Ann(\N)) = 6$ so we have $\N\N = \Ann(\N) = C(\N)$. It shows
  $\N$ is 2-step nilpotent.
\end{proof}

In this case, the character matrix of $\N$ in the basis $\{x_1, x_2,x_3,z_1,z_2,z_3\}$ is given by:

\[\begin{pmatrix} A & 0 \\ 0 & 0 \end{pmatrix},
\]
where A is a $3\times3$-matrix defined by the structure constants
$x_ix_j= \sum_k c_{ij}^kz_k$, $1\leq i,j,k\leq 3$, and $B$ has the matrix:
\[\begin{pmatrix} 0 & 0 & 0 & 1& 0 & 0 \\ 0 & 0 & 0 & 0& 1 & 0 \\ 0 &
  0 & 0 & 0& 0 & 1 \\ 1 & 0 & 0 & 0& 0 & 0 \\ 0 & 1 & 0 & 0& 0 & 0 \\
  0 & 0 & 1 & 0 & 0 & 0 \end{pmatrix}.\]

Since $B(x_ix_j,x_r) = B(x_i,x_jx_r)=B(x_j,x_rx_i)$ then one has
$c_{ij}^r = c_{jr}^i = c_{ri}^j$, $1\leq i,j,k\leq 3$.

Next, we give some simple properties for symmetric Novikov algebras as follows:

\begin{prop}\label{prop4.28}
Let $\N$ be a symmetric non-commutative Novikov algebra. If $\N$ is reduced then
\[
3\leq \dim(\Ann(\N))\leq \dim(\N\N) \leq\dim(\N) -3.\]
\end{prop}

\begin{proof}
  By Lemma \ref{lem4.23}, $\N\N\subset C(\N)$. Moreover, $\N$ non-commutative
  implies that $\g(\N)$ is non-Abelian and by \cite{PU07}, $\dim
  \left([\N,\N] \right) \geq 3$. Therefore, $\dim C(\N)
  \leq\dim(\N) -3$ since $C(\N) = [\N,\N]^\bot$. Consequently,
  $\dim(\N\N) \leq\dim(\N) -3$ and then $\dim(\Ann(\N))\geq 3$.
\end{proof}

\begin{cor}\label{cor4.29}
  Let $\N$ be a symmetric non-commutative Novikov algebra of dimension
  7. If $\N$ is 2-step nilpotent then $\N$ is not reduced.
\end{cor}
\begin{proof}
  Assume that $\N$ is reduced then $\dim(\Ann(\N))= 3$ and
  $\dim(\N\N)=4$. It implies that there must have a nonzero element
  $x\in \N\N$ such that $x\N\neq \{0\}$ and then $\N$ is not 2-step
  nilpotent.
\end{proof}

Now, we give a more general result for symmetric Novikov algebra of
dimension 7 as follows:

\begin{prop}\label{prop4.30}
  Let $\N$ be a symmetric non-commutative Novikov algebra of dimension
  7. If $\N$ is reduced then there are only two cases:
\begin{enumerate}
    \item $\N$ is 3-step nilpotent and indecomposable.
    \item $\N$ is decomposable by $\N =\CC x\oplusp \N_6$, where $x^2=x$ and $\N_6$ is a symmetric non-commutative Novikov algebra of dimension 6.
     \end{enumerate}
\end{prop}
\begin{proof}
  Assume that $\N$ is reduced then $\dim(\Ann(\N))= 3$, $\dim(\N\N)=4$
  since $\Ann(\N)\subset \N\N$ and $\Ann(\N)= (\N\N)^\bot$. By
  \cite{Bou59}, $\Ann(\N)$ is totally isotropic, then there exist a
  totally isotropic subspace $V$ and a nonzero $x$ of $\N$ such that
\[\N =\Ann(\N)\oplus \CC x\oplus V,
\] where $\Ann(\N)\oplus V$ is non-degenerate, $\ B(x,x)\neq 0$ and $
x^\bot = \Ann(\N)\oplus V$. As a consequence, $\Ann(\N)\oplus \CC x =
(\Ann(\N))^\bot = \N\N$.

Consider the left-multiplication operator $L_x: \CC x\oplus V
\rightarrow \Ann(\N)\oplus \CC x$, $L_x(y) = xy, \ \forall y\in\CC
x\oplus V$. Denote by $p$ the projection $\Ann(\N)\oplus \CC x
\rightarrow \CC x$.

\begin{itemize}
\item If $p\circ L_x = 0$ then $(\N\N)\N = x\N\subset
  \Ann(\N)$. Therefore, $((\N\N)\N)\N = \{0\}$. That implies $\N$ is
  3-nilpotent. If $\N$ is decomposable then $\N$ must be 2-step
  nilpotent. This is in contradiction to Corollary \ref{cor4.29}.

\item If $p\circ L_x \neq 0$ then there is a nonzero $y\in \CC x\oplus
  V$ such that $xy = ax + z$ with $0\neq a\in\CC$ and
  $z\in\Ann(\N)$. In this case, we can choose $y$ such that $a=1$. It
  implies that $(x^2)y = x(xy) = x^2$.

  If $x^2 = 0$ then $0=B(x^2, y) = B(x,xy) = B(x,x)$. This is a
  contradiction. Therefore, $x^2 \neq 0$. Since $x^2\in\Ann(\N)\oplus
  \CC x$ then $x^2 = z' + \mu x$, where $z'\in \Ann(\N)$ and
  $\mu\in\CC$ must be nonzero. By setting $x': = \frac{x}{\mu}$ and
  $z'' = \frac{z'}{\mu^2}$, we get $(x')^2 = z'' + x'$. Let $x_1: =
  (x')^2$, one has:
\[x_1^2=(x')^2(x')^2 = (z'' + x')(z'' + x') =x_1.\]

Moreover, for all $t=\lambda x + v\in \CC x\oplus V$, we have
$t(x^2)=(x^2)t = x(xt) = \lambda \mu (x^2)$. It implies that $\CC x^2
= \CC x_1$ is an ideal of $\N$.

Since $B(x_1,x_1)\neq 0$, by Lemma \ref{lem4.6} one has $\N =\CC
x_1\oplusp (x_1)^\bot$. Certainly, $(x_1)^\bot$ is a symmetric
non-commutative Novikov algebra of dimension 6.
\end{itemize}

\end{proof}

\begin{prop}\label{prop4.30}
  Let $\N$ be a symmetric Novikov algebra. If $\g(\N)$ or $\J(\N)$ is
  reduced then $\N$ is reduced.
\end{prop}
\begin{proof}
  Assume that $\N$ is not reduced then there is a nonzero $x\in
  \Ann(\N)$ such that $B(x,x) = 1$. Since $[x,\N] = [x,\N]_+ = 0$ then
  $\g(\N)$ and $\J(\N)$ are not reduced.
\end{proof}
\begin{cor}
  Let $\N$ be a symmetric Novikov algebra. If $\g(\N)$ is reduced then
  $\N$ must be 2-step nilpotent.
\end{cor}
\begin{proof}
  Since $\g(\N)$ is reduced then $\Ann(\N)\subset \N\N$. On the other
  hand, $\dim(C(\N)) = \dim([\N,\N]) = \frac{1}{2}\dim(\N)$ so
  $\dim(\Ann(\N)) = \dim(\N\N)$. Therefore, $\Ann(\N) = \N\N$ and $\N$
  is 2-step nilpotent.
\end{proof}
\begin{ex} \label{ex4.32} By Example \ref{ex4.2}, every 2-step
  nilpotent algebra is Novikov then we will give here an example of
  symmetric non-commutative Novikov algebras of dimension 7 which is
  3-step nilpotent. Let $\N = \CC x\oplus \N_6$ be a 7-dimensional
  vector space, where $\N_6$ is the symmetric Novikov algebra of
  dimension 6 in Example \ref{ex4.25}. Define the product on $\N$ by
\[ xe_4 = e_4x = e_1, e_4e_4 = x, e_4e_5 = e_3, e_5e_6 = e_1, e_6e_4 = e_2,
\]
and the symmetric bilinear form $B$ defined by
\begin{eqnarray*} B(x,x) = B(e_1,e_4) = B(e_2,e_5) = B(e_3,e_6) = 1 \\
  B(e_4,e_1) = B(e_5,e_2) = B(e_6,e_3) = 1,\\ 0 \ \text{ otherwise.}
\end{eqnarray*}
\end{ex}

Note that in above Example, $\g(\N)$ is not reduced since $x\in
C(\N)$.

\section{Appendix}

\begin{lem}
  Let $(V,B)$ be a quadratic vector space, $C$ be an invertible
  endomorphism of $V$ such that
\begin{enumerate}
    \item $B(C(x),y) = B(x,C(y)), \forall x,y\in V$.
    \item $3C - 2C^2 = \Id$
\end{enumerate}
Then there is an orthogonal basis $\{e_1, ..., e_n\}$ of $B$ such that
$C$ is diagonalizable with eigenvalues $1$ and $\frac{1}{2}$.
\end{lem}

\begin{proof}
  Firstly, one has (2) equivalent to $C(C-\Id) =
  \frac{1}{2}(C-\Id)$. Therefore, if $x$ is a vector in $V$ such that
  $C(x) - x \neq 0$ then $C(x) - x$ is an eigenvector with respect to
  eigenvalue $\frac{1}{2}$. We prove the result by induction on
  $\dim(V)$. If $\dim(V) = 1$, let $\{e\}$ be a orthogonal basis of
  $V$ and assume $C(e) = \lambda e$ for some $\lambda \in \CC$. Then
  by (2) one has $\lambda = 1$ or $\lambda = \frac{1}{2}$.

Assume that the result is true for quadratic vector spaces of
dimension $n$, $n \geq 1$. Assume $\dim(V) = n+1$. If $C = \Id$ then
the result follows. If $C \neq \Id$ then there exists $x\in V$ such
that $C(x) -x \neq 0$. Let $e_1: = C(x) - x$ then $C(e_1) =
\frac{1}{2}e_1$.

If $B(e_1, e_1) = 0$ then there is $e_2\in V$ such that $B(e_2,e_2) =
0$, $B(e_1, e_2)=1$ and $V = \spa\{e_1,e_2\}\oplusp V_1$, where $V_1 =
\spa\{e_1,e_2\}^\bot$. Since $\frac{1}{2} = B(C(e_1),e_2) =
B(e_1,C(e_2))$ one has $C(e_2)= \frac{1}{2}e_2 + x + \beta e_1$, where
$x\in V_1, \beta\in\CC$. Let $f_1: = C(e_2) - e_2 = -\frac{1}{2}e_2 +
x+ \beta e_1$ then $C(f_1) = \frac{1}{2}f_1$ and $B(e_1,f_1) =
-\frac{1}{2}$. If $B(f_1, f_1) \neq 0$ then let $e_1 := f_1$. If
$B(f_1, f_1) = 0$ then let $e_1 := e_1 + f_1$. In the both cases, we
have $B(e_1, e_1) \neq 0$ and $C(e_1) = \frac{1}{2}e_1$. Let $V = \CC
e_1 \oplusp e_1^\bot$ then $e_1^\bot$ is non-degenerate, $C$ maps
$e_1^\bot$ into itself. Therefore the result follows the induction
assumption.
\end{proof}

\bibliographystyle{amsxport}

\begin{bibdiv}
\begin{biblist}

\bib{AB10}{article}{
   author={Ayadi, I.},
   author={Benayadi, S.},
   title={Symmetric Novikov superalgebras},
   journal={Journal of Mathematical Physics},
   volume={51},
   number={2},
   date={2010},
   pages={023501},

}

\bib{Alb49}{article}{
   author={Albert, A. A.},
   title={A theory of trace-admissible algebras},
   journal={Proceedings of the National Academy of Sciences of the United States of America},
   volume={35},
   number={6},
   date={1949},
   pages={317--322},

}
\bib{BB08}{article}{
   author={Baklouti, A.},
   author={Benayadi, S.},
   title={Pseudo-euclidean Jordan algebras},
   journal={arXiv:0811.3702v1}
}

\bib{BB99}{article}{
     author={Benamor, H.},
   author={Benayadi, S.}
   title={Double extension of quadratic Lie superalgebras},
   journal={Communications in Algebra},
   volume={27},
   number={1},
   date={1999},
   pages={67 -- 88},

}
\bib{BM01}{article}{
     author={Bai, C.},
      author={Meng, D.},
   title={The classification of Novikov algebras in low dimensions},
   journal={J. Phys. A: Math. Gen.},
   volume={34},
   date={2001},
   pages={1581 -- 1594},}

\bib{BM02}{article}{
     author={Bai, C.},
      author={Meng, D.},
   title={Bilinear forms on Novikov algebras},
   journal={Int. J. Theor. Phys.},
   volume={41},
   number={3},
   date={2002},
   pages={495 -- 502},}

\bib{BMH02}{article}{
     author={Bai, C.},
     author={Meng, D.},
     author={He, L},
   title={On fermionic Novikov algebras},
   journal={J. Phys. A: Math. Gen.},
   volume={35},
   number={47},
   date={2002},
   pages={10053 -- 10063},}

\bib{BN85}{article}{
     author={Balinskii, A.A.},
     author={Novikov, S. P.},
     title={Poisson brackets of hydrodynamic type, Frobenius algebras and Lie algebras},
   journal={Dokl. Akad. Nauk SSSR},
   volume={283},
   number={5},
   date={1985},
   pages={1036 -- 1039},}

\bib{Bor97}{article}{
     author={Bordemann, M.}
     title={Nondegenerate invariant bilinear forms on nonassociative algebras},
   journal={Acta Math. Univ. Comenianae},
   volume={LXVI},
   number={2},
   date={1997},
   pages={151 -- 201},
}


\bib{Bou59}{book}{
   author={Bourbaki, N.},
   title={Eléments de Mathématiques. Algèbre, Formes sesquilinéaires et formes quadratiques},
   volume={Fasc. XXIV, Livre II},
   publisher={Hermann},
   place={Paris},
   date={1959},
   pages={},
}



\bib{Bur06}{article}{
   author={Burde, D.},
   title={Classical $r$-matrices and Novikov algebras},
   journal={Geometriae Dedicata},
   volume={122},
   date={2006},
   pages={145--157},

}


\bib{CM93}{book}{
   author={Collingwood, D. H.},
   author={McGovern, W. M.},
   title={Nilpotent Orbits in Semisimple Lie. Algebras},
   publisher={Van Nostrand Reihnhold Mathematics Series},
   place={New York},
   date={1993},
   pages={186},
}

\bib{DPU10}{article}{
   author={Duong, M.T.},
   author={Pinczon, G.},
   author={Ushirobira, R.},
   title={A new invariant of quadratic Lie algebras},
   journal={arXiv:1005.3970v2},

}

\bib{Duo10}{article}{
   author={Duong, M. T.},
   journal={Thèse de l'Université de Bourgogne},
   date={2010},
   pages={in preparation},

}


\bib{FK94}{book}{
   author={Faraut, J.},
   author={Koranyi, A.},
   title={Analysis on symmetric cones},
   publisher={Oxford Mathematical Monographs},
   date={1994},
   pages={382},
}


\bib{FS87}{article}{
   author={Favre, G.},
   author={Santharoubane, L.J.},
   title={Symmetric, invariant, non-degenerate bilinear form on a Lie algebra},
   journal={Journal of Algebra},
   volume={105},
   date={1987},
   pages={451--464},

}
\bib{GD79}{article}{
     author={Gel'fand, I. M.},
     author={Dorfman, I. Ya.},
     title={Hamiltonian operators and algebraic structures related to them},
   journal={Funct. Anal. Appl},
   volume={13},
   number={4},
   date={1979},
   pages={248 -- 262},}

\bib{Jac51}{article}{
   author={Jacobson, N.},
   title={General representation theory of Jordan algebras},
   journal={Trans. Amer. Math. Soc},
   volume={70},
   date={1951},
   pages={509--530},

}

\bib{Kac85}{book}{
   author={Kac, V.},
   title={Infinite-dimensional Lie algebras},
   publisher={Cambrigde University Press},
   place={New York},
   date={1985},
   pages={xvii + 280 pp}

}

\bib{MR85}{article}{
  author={Medina, A.},
  author={Revoy, Ph.},
  title={Algèbres de Lie et produit scalaire invariant},
  journal={Ann. Sci. École Norm. Sup.},
  volume={4},
  date={1985},
  pages={553 -- 561},

}
\bib{Ova07}{article}{
  author={Ovando, G.},
  title={Two-step nilpotent Lie algebras with ad-invariant metrics and a special kind of skew-symmetric maps},
  journal={J. Algebra and its Appl.},
  volume={6},
  date={2007},
  number={6},
  pages={897 -- 917},

}

\bib{PU07}{article}{
   author={Pinczon, Georges},
   author={Ushirobira, Rosane},
   title={New Applications of Graded Lie Algebras to Lie Algebras, Generalized Lie Algebras, and Cohomology},
   journal={Journal of Lie Theory},
   volume={17},
   date={2007},
   number={3},
   pages={633 -- 668},

}
\bib{Sch55}{article}{
   author={Schafer, R. D.},
   title={Noncommutative Jordan algebras of characteristic 0},
   journal={Proceedings of the American Mathematical Society},
   volume={6},
   date={1955},
   number={3},
   pages={472 -- 475},

}
\bib{Sch61}{book}{
   author={Schafer, R. D.},
   title={An Introduction to Nonassociative Algebras},
   publisher={Academic Press},
   place={New York},
   date={1966},
  }
\bib{ZC07}{article}{
   author={Zhu, F.},
   author={Chen, Z.},
   title={Novikov algebras with associative bilinear forms},
   journal={Journal of Physics A: Mathematical and Theoretical},
   volume={40},
   date={2007},
   number={47},
   pages={14243--14251},


}


\end{biblist}
\end{bibdiv}

\end{document}